\documentclass[11pt]{article}

\usepackage[english]{babel}
\usepackage{amsmath, amsthm, amssymb, amscd}
\usepackage{color}
\usepackage{graphicx}
\usepackage{natbib}

\usepackage{hyperref}
\usepackage{dsfont}

\usepackage{authblk} 

\newtheorem{thm}{Theorem}
\newtheorem{lem}{Lemma}

\newcommand{\argmin}{\mathop{\rm arg\min}}

\newcommand{\Deq}{\overset{\mathcal{D}}{=}}
\newcommand{\Darrow}{\overset{\mathcal{D}}{\longrightarrow}}
\newcommand{\Parrow}{\overset{\mathbb{P}}{\longrightarrow}}

\newcommand{\Cov}{\operatorname{Cov}}
\newcommand{\Var}{\operatorname{Var}}
\newcommand{\SKK}{\mathrm{SKK}}
\newcommand{\IHS}{\mathrm{IHS}}
\newcommand{\Samsee}{\mathrm{SAMSEE}}
\newcommand{\ISE}{\mathrm{ISE}}
\newcommand{\iid}{\overset{i.i.d.}{\sim}}

\begin{document}

\title{{\huge Threshold Selection\\ in Univariate Extreme Value Analysis}}
\author[1]{Laura Fee Schneider\thanks{laura-fee.schneider@mathematik.uni-goettingen.de}}
\author[1]{Andrea Krajina\thanks{akrajina@gmail.com}}
\author[1]{Tatyana Krivobokova\thanks{tkrivob@gwdg.de}}
\affil[1]{Institute for Mathematical Stochastics, University of G\"{o}ttingen}
\date{ }
\maketitle

\vspace*{-1cm}

\begin{abstract}
Threshold selection plays a key role for various aspects of statistical inference of rare events. Most classical approaches tackling this problem for heavy-tailed distributions crucially depend on tuning parameters or critical values to be chosen by the practitioner. To simplify the use of automated, data-driven threshold selection methods, we introduce two new procedures not requiring the manual choice of any parameters. The first method measures the deviation of the log-spacings from the exponential distribution and achieves good performance in simulations for estimating high quantiles. The second approach smoothly estimates the asymptotic mean square error of the Hill estimator and performs consistently well over a wide range of distributions.\\
The methods are compared to existing procedures in an extensive simulation study and applied to a dataset of financial losses, where the underlying extreme value index is assumed to vary over time. This application strongly emphasizes the importance of solid automated threshold selection.
\end{abstract}

\paragraph{AMS 2010 Subject Classification:}
Primary 62G32; secondary 62G05, 62F12, 97M30

\paragraph{Keywords:} Extreme value statistics, peak-over-threshold approach, power laws, Hill estimator, tuning parameter selection, bias estimation

\section{Introduction}

Extreme value analysis of heavy-tailed distributions is an important model in various applications. In seismology and climatology, for example, statistics of extremes is used to study earthquakes \citep{Beirlant-earthquake} or heavy precipitation \citep{Naveau-rain}. Another important field of research is analysing high financial losses, which becomes particularly interesting if the losses depend on covariates \citep{Valerie-loss, Hambuckers}. In this situation an automated threshold selection procedure could bring additional benefits by enabling the selection of the threshold depending on a covariate. We will discuss this possibility in more detail in Section \ref{S:app}.\\
To mathematically investigate the behaviour of heavy tails, we consider random variables from the domain of attraction (DoA) of a Fr\'echet distribution. Let $X_1,\dots,X_n $ be independent identically distributed (i.i.d.) random variables with distribution function $F$, where $F$ is in the DoA of an extreme value distribution (evd)
$G_{\gamma}$ with extreme value index $\gamma>0$. This means there exist sequences $a_n>0$ and $b_n$ real, s.t.
\begin{align*}
\lim_{n\rightarrow\infty} F^n(a_n x+b_n) = G_{\gamma}(x) : = \exp\left( -x^{-1/\gamma}\right).
\end{align*}
In this situation the following first order condition holds,
\begin{align}\label{E:first order condition}
  \underset{t\rightarrow\infty}{\lim}\ \frac{1-F(tx)}{1-F(t)}=x^{-1/\gamma},
\end{align}
i.e.\ the survival function $1-F$ is regularly varying with index $-1/\gamma$. Distributions fulfilling this condition are called Pareto-type distributions, because they only differ from the Pareto distribution by a slowly varying function $\ell_F(x)$, i.e.\ $1-F(x)=x^{-1/\gamma}\ell_F(x)$.\\
We can interpret the quotient in (\ref{E:first order condition}) as a conditional probability, and it follows directly that
\begin{align}
 \frac{X_1}{t}\,\Big|\, X_1>t\ \overset{\mathcal{D}}{\longrightarrow}\ & P, \text{ as } t\rightarrow\infty \text{ and } P\sim \mathrm{Pareto}\left(1,\frac{1}{\gamma}\right), \nonumber \\
 \log\left(\frac{X_1}{t}\right)\Big|\,X_1>t\ \overset{\mathcal{D}}{\longrightarrow}\ & E, \text{ as } t\rightarrow\infty \text{ and } E\sim \mathrm{Exp}\left(\frac{1}{\gamma}\right). \label{E:PoT}
\end{align}
Thus, for a sufficiently large threshold $t$ the data above this threshold can be modelled by a Pareto or an exponential distribution. In this article we concentrate on the exponential approximation and utilize it for inference on the extreme value index. It is common to consider the threshold $t=X_{(n-k,n)}$ and choose the sample fraction $k$ instead of $t$, where $X_{(1,n)}\leq\cdots\leq X_{(n,n)}$ denote the order statistics of a sample of size $n$. In this case, a natural estimator for $\gamma$ under the exponential approximation of the log-spacings $Y_{(i,k)}:=\log(X_{(n-i+1,n)})-\log(X_{(n-k,n)})$ is their mean, the Hill estimator \citep{Hill},
\begin{equation}\label{E:Hill}
\hat{\gamma}_k := \frac{1}{k}\sum_{i=1}^k \log\left( \frac{X_{(n-i+1,n)}}{X_{(n-k,n)}} \right) = \frac{1}{k}\sum_{i=1}^k Y_{(i,k)}. 
\end{equation}
The Hill estimator is still among the most popular and well-known estimators for the extreme value index, although its sample path as a function in $k$ can be highly unstable and estimation therefore crucially depends on the choice of the sample fraction $k$.
This dependence highlights the difficulties in estimating $\gamma$: even from univariate i.i.d.\ observations from $F\in\mathrm{DoA}(G_{\gamma})$, estimation is hard, since only few observations contain information about the extreme value distribution $G_{\gamma}$. To select a threshold above which the data can be used for statistical inference about the tail is one of the most fundamental problems in the field of extreme value analysis.\\

Due to the importance of this task, the appropriate choice of the threshold has been discussed extensively in extreme value research over the last decades, and suggested solutions cover a variety of methodologies. We give a short summary on different types of approaches and stress the specific difficulties that arise. We mainly concentrate on methods we compare in our simulation study in Section \ref{S:Simulation}. More comprehensive reviews about threshold selection can be found in \cite{review} and \cite{DeyYan}.\\
One basic concept in threshold selection is data visualisation, which is also discussed more deeply in \cite{KratzResnick} and \cite{DreesHill}. Popular graphical diagnostics used in this context are the Zipf plot, Hill plot, QQ-plot or the mean-excess plot to name a few. A major drawback of these methods is their subjectivity due to the necessarily personal interpretation of the plot. Further, it is a burden to choose each threshold manually, especially in high dimensional settings or when analysing many samples. 
Easier ways to select the sample fraction are rules-of-thumb such as using the upper 10\% of the data \citep{DuMouchel} or $k=\sqrt{n}$ \citep{Ferreira-quantiles}. However, these suggestions are neither theoretically justified nor data driven. 
\cite{Reiss-Thomas} present a procedure that tries to find a region of stability among the estimates of the extreme value index. Their method depends on a tuning parameter, whose choice is further analysed in \cite{Neves-RT}. To our knowledge no theoretical analysis exists for this approach.\\
Besides these and similar heuristic approaches, there is a class of theoretically motivated procedures that target the optimal sample fraction for specific estimation tasks, such as quantile estimates \citep{Ferreira-quantiles}, estimation of high probabilities \citep{HallWeissman} or the Hill estimator, see below. We also mention two other methodologies. First, there are suggestions that utilize comparing the empirical distribution to the fitted generalized Pareto distribution (GPD) via goodness-of-fit tests \citep{Bader} or by minimizing the distance between them \citep{Pickands, GonzaloOlmo, Clauset}, where the latter approach is theoretically analysed by \cite{Drees_Clauset}. Further, \cite{Goegebeur} propose a family of kernel statistics to test for exponentiality in order to select a threshold. \\
Of particular interest to us are methods that aim to estimate the sample fraction $k_{\mathrm{opt}}$ which minimizes the asymptotic mean square error (AMSE) of the Hill estimator. To construct an estimator for $k_{\mathrm{opt}}$, \cite{DreesK} utilize the Lepskii method and an upper bound on the maximum random fluctuation of $\hat{\gamma}_k$ around $\gamma$. To apply their approach it is necessary to choose several tuning parameters and to obtain consistent initial estimates for $\gamma$ and a second order parameter $\rho$. They recommend specific choices of the parameters based on a numerical study and we employ their proposals in our simulations. However, the choice of these parameters is not data-driven.
In \cite{Guillou}, a test statistic $Q_k$ is constructed based on an accumulation of log-spacings, which takes values around 1 as long as the bias of the Hill estimator is not significantly large. Their statistic depends on a tuning parameter as well, and a critical value to test $Q_k$ against has to be chosen. Again we adopt the parameter choice suggested in their simulation study.
\cite{Danielsson} introduce a double bootstrap approach to estimate the optimal sample fraction. They need to choose the number of bootstrap samples and a parameter $n_1$. For $n_1$, a data-driven but computationally expensive selection method is provided, where the whole bootstrap procedure is repeated for various possible values of $n_1$.
Another estimator for $k_{\mathrm{opt}}$ is given by \cite{Beirlant02}, which employs least squares estimates from an exponential regression approach. The method depends on an estimate for $\rho$ and a sample fraction $k_0$. To avoid the choice of $k_0$ they suggest taking the median of the estimates over a range of values, e.g.\ $k_0\in\{3,\dots,n/2\}$.
A different approach is taken by \cite{Goegebeur}, who use the properties of a test statistic regarding bias estimation to construct an estimator for the AMSE$/\gamma$ and minimize it with respect to $k$. If one fixes $\rho=-1$, as they suggest in their simulations chapter, there is no further tuning parameter to be chosen. However, no result about consistency of $\hat{k}$ in the sense of $\hat{k}/k_{\mathrm{opt}}\overset{\mathbb{P}}{\rightarrow} 1$ is known in contrast to the approaches in \cite{DreesK}, \cite{Guillou}, \cite{Danielsson} and \cite{Beirlant02}.\\

In this paper we contribute to the problem of threshold selection by introducing two new methods. The first one presented in Section \ref{S:SKK} is inspired by the idea of testing the exponential approximation. We estimate the integrated square error (ISE) of the exponential density under the assumption that the log-spacings are indeed exponentially distributed. The error functional we obtain, denoted as inverse Hill statistic (IHS), is very easy to compute and does not depend on any tuning parameters. Since this criterion is variable for small $k$, it can be additionally smoothed to improve the performance. The minimizing sample fraction of IHS is asymptotically smaller than $k_{\mathrm{opt}}$, as it is stricter against deviation from the exponential approximation. This estimator performs remarkably well for adaptive quantile estimation on finite samples, as illustrated in our simulation study.\\
In our second approach we suggest a smooth estimator for the AMSE of the Hill estimator, called SAMSEE (smooth AMSE estimator). This estimator is constructed by a preliminary estimate of $\gamma$ using the generalized Jackknife approach in \cite{Gomes01} and a bias estimator for the Hill estimator introduced in Section \ref{S:estimateMSE}. By minimizing SAMSEE we estimate the optimal sample fraction $k_{\mathrm{opt}}$. For estimation, the choice of a large sample fraction $K$ is necessary, for which we present a data-driven selection procedure in Section \ref{S:estimateMSE}. 
SAMSEE utilizes the idea of fixing $\rho=-1$, which is justified by good performance in simulations and leads to a simpler and more robust estimator. However, the estimator can also be adjusted to any $\rho$ by including a consistent estimator $\hat{\rho}$, as described in Section \ref{S:rho not -1}.\\
After introducing our two novel threshold selection methods in Sections \ref{S:SKK} and \ref{S:estimateMSE} we compare these methods to various other approaches in an numerical analysis in Section \ref{S:Simulation}. In Section \ref{S:app} the importance of automated threshold selection procedures is illustrated in an application, where we  non-parametrically estimate an extreme value index that varies over time.
The proof of Theorem \ref{T:Bias}, which describes the asymptotic behaviour of our bias estimator, and auxiliary theoretical results can be found in Appendix \ref{S:proofs}.

\section{IHS -- The inverse Hill statistic}\label{S:SKK}

In this section we introduce the first threshold selection procedure by analysing the integrated square error (ISE) between the exponential density $h_{\gamma}$ and its parametric estimator $h_{\hat{\gamma}_k}$ employing the Hill estimator,
\begin{equation}
\ISE(k):= \int \left(h_{\gamma}(x) -h_{\hat{\gamma}_k}(x)\right)^2\mathrm{d}x
=\frac{1}{2\gamma}-\frac{2}{\gamma +\hat{\gamma}_k} + \frac{1}{2\hat{\gamma}_k} . \nonumber
\end{equation} 
The first term of ISE is constant and thus plays no role for selecting $k$. The last term of ISE is known, but the second term is not.
Therefore, we cannot minimize ISE directly. Instead, we want to estimate and minimize its expectation under the exponential approximation. This is based on the idea of considering the hypothesis $H_0$ that the log-spacings $Y_{(i,k)}$ are indeed exponentially distributed. Under $H_0$ the Hill estimator is gamma distributed, see Lemma \ref{L:Hill-dist}, and the mean of ISE (MISE) can be calculated explicitly. We observe that MISE is a decreasing function in $k$  under the exponential approximation,
\begin{equation}\label{E:MISE_terms}
\mathrm{MISE}(k)-\frac{1}{2\gamma}:= \mathbb{E}_{H_0}[\ISE(k)]-\frac{1}{2\gamma} =-\frac{1}{\gamma}C(k)+\frac{k}{2(k-1)\gamma} ,
\end{equation}
where $C(k):=2\exp(k) k^k \Gamma(1-k,k)$ and $\Gamma(a,b)$ denotes the upper incomplete gamma function. The function $C(k)$ converges to 1 very fast, s.t.\ we obtain 
\begin{equation}
\mathbb{E}_{H_0}\left[ \frac{2}{\gamma+\hat{\gamma}_k}\right] \approx \frac{1}{\gamma} =  \mathbb{E}_{H_0}\left[\frac{k-1}{k\hat{\gamma}_k}\right] .\nonumber
\end{equation} 
This provides us with an unbiased estimator for the first term in (\ref{E:MISE_terms}) under $H_0$. However, due to the high variability for small $k$, we instead want to find an estimator of the form $w/\hat{\gamma}_k$ for some $w$ depending on $k$ that minimizes the MSE under the exponential approximation. To do so, we approximate its MSE in the following way,
\begin{align}\label{E:MSE_MISE}
\mathbb{E}_{H_0}\left[\left(\frac{w}{\hat{\gamma}_k}-\frac{2}{\hat{\gamma}_k+\gamma}\right)^2\right] \approx \frac{w^2k^2}{\gamma^2(k-1)(k-2)} -\frac{2wk}{\gamma^2(k-1)} + \frac{1}{\gamma^2}.
\end{align}
The approximation depends on similar functions as $C(k)$, which quickly become constant. The MSE in (\ref{E:MSE_MISE}) is minimized for $w=(k-2)/k$. Thus, we suggest the inverse Hill statistic
\begin{align*}
\IHS(k) := \frac{1}{2\hat{\gamma}_k} -\frac{k-2}{\hat{\gamma}_k k} = \frac{4-k}{2\hat{\gamma}_k k} 
\end{align*}
to estimate $\mathrm{MISE}(k) -(2\gamma)^{-1}$ and the threshold selected via minimizing IHS,
\begin{align*}
\hat{k}_{\IHS} := \argmin_{1<k<n}\ \IHS(k).
\end{align*}
By minimizing IHS we select a sample fraction where IHS starts increasing and contradicts $H_0$ by behaving contrarily to MISE under the exponential approximation. This criterion can be compared to hypothesis testing with a large significance level $\alpha$, which implies seeking high confidence when deciding to not reject $H_0$. Further properties of $\hat{k}_{\IHS}$ are analysed theoretically in Section \ref{S:SKK-theory} and for finite samples in a numerical study in Section \ref{S:Simulation}. \\
Note that the performance of IHS depends on the bias of the Hill estimator being positive and increasing, see Section \ref{S:SKK-theory}. However, the bias can be negative for some non-standard distributions. In case of a negative bias, we instead suggest to use,
\begin{align*}
\IHS^{-}(k) := \frac{4+k}{2\hat{\gamma}_k k}\quad \text{and} \quad
\hat{k}_{\IHS^{-}} := \argmin_{1<k<n}\ \IHS^{-}(k).
\end{align*}
The two cases can easily be distinguished by analysis of the Hill estimator for large $k$. Both $\IHS$ and $\IHS^{-}$ are justified by asymptotic results in Section \ref{S:SKK-theory}.\\
Figure \ref{F:Hill-SKK} illustrates that IHS is highly varying for small $k$, which makes automatic threshold choices more variable. To control this problematic behaviour we smooth the IHS. More specifically, we want to estimate $\mathbb{E}[\IHS]$ by considering the regression problem 
\begin{align*}
\IHS(k) = \mathbb{E}[\IHS](k) + \sigma\epsilon_k,\ \ k=1,\dots,n,
\end{align*} 
where $\sigma>0$ and $\mathbb{E}[\epsilon_k]=0$. 
Due to the structure of the Hill estimator, the random variables $\epsilon_k$ are highly dependent, which needs to be taken into account in estimation. In our simulations, we apply a Bayesian non-parametric procedure introduced by \citet{SerraKR} which simultaneously estimates mean and covariance and is available in the R-package {\it eBsc}. The approach provides a smooth estimator for the expectation of IHS -- denoted as sIHS -- comprising less variation for small $k$. This way we can improve the performance by selecting a more suitable threshold, as illustrated in Figure \ref{F:Hill-SKK}. Of course, one can also use other smoothing procedures suitable for dependent data \citep{Opsomer, Krivobokova07, Lee10}.\\
\begin{figure}[h]
\centering
\includegraphics[width=0.48\textwidth, height=0.28\textheight]{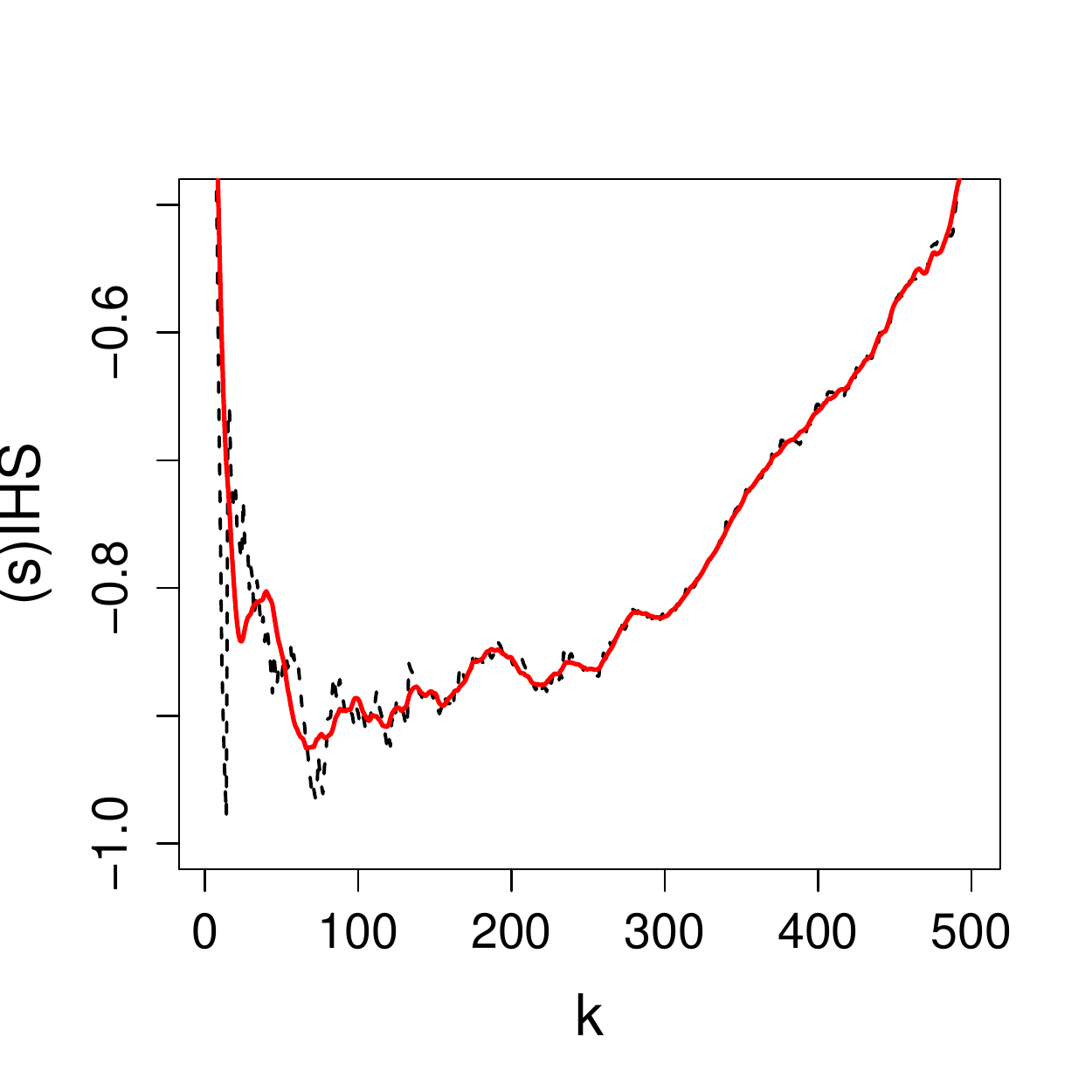} \quad 
\includegraphics[width=0.48\textwidth, height=0.28\textheight]{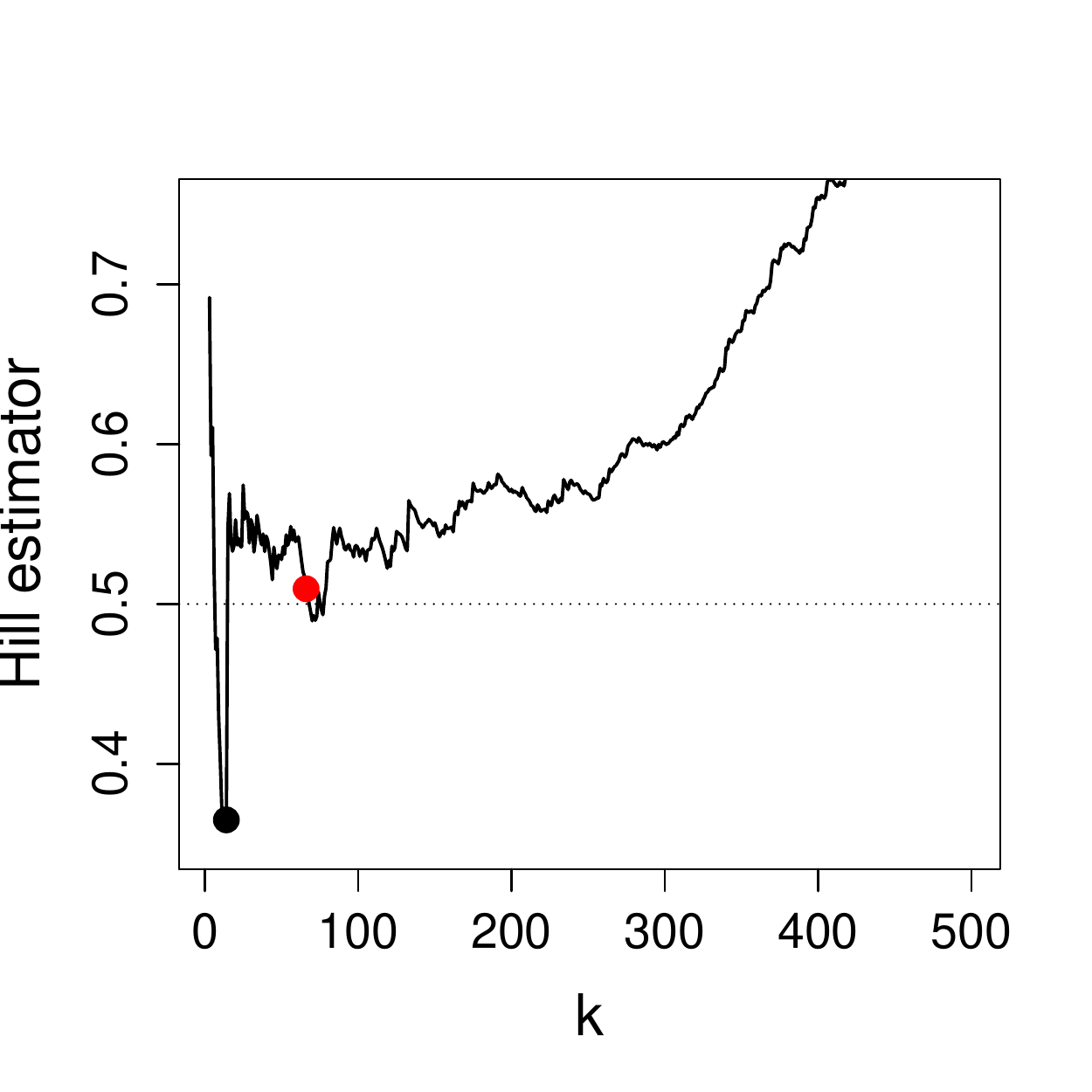}
\caption{On the left, the IHS (dashed) and sIHS (red) are plotted for a Fr\'echet(2) sample of size $n=500$. On the right the Hill plot for the same sample with the minimizing $k$ of IHS (black) and sIHS (red) is shown, where the dotted line marks the true value of $\gamma=1/2$. }
\label{F:Hill-SKK}
\end{figure}

We finally want to remark on the relation between IHS and ISE, which is given by
\begin{equation}\label{E:SKKvsISE}
\IHS +\frac{1}{2\gamma}= \ISE  +\frac{2}{k\hat{\gamma}_k} +\frac{\hat{\gamma}_k- \gamma}{\hat{\gamma}_k(\hat{\gamma}_k+\gamma)}.
\end{equation}
This equation points out that minimizing IHS does not minimize ISE, as IHS takes an additional bias term into account. If the bias of the Hill estimator is positive, IHS selects smaller $k$ (larger thresholds) than ISE. This is not surprising, because we estimate the expectation of the ISE under the hypothesis that the exponential approximation holds. This is a much more conservative error functional, meaning it is more strict against deviation from the exponential distribution.\\
In conclusion, with IHS we do not aim to estimate $k_{\mathrm{opt}}$ but to find a sample fraction where we can be very certain that the exponential approximation still holds. The impact of this consideration is illustrated in simulations and an application in Sections \ref{S:Simulation} and \ref{S:app}.

\subsection{Theorectical analysis of IHS}\label{S:SKK-theory}

In order to understand the IHS asymptotically we consider the second order condition,
\begin{equation}\label{E:second order condition}
\underset{t\rightarrow\infty}{\lim} \frac{\frac{U(tx)}{U(t)}-x^{\gamma}}{A(t)} =x^{\gamma}\frac{x^{\rho}-1}{\rho},
\end{equation}
for $x>0$ and with second order parameter $\rho<0$. Here, $A(t)$ denotes a function converging to zero as $t$ goes to infinity and $|A|$ is regularly varying with index $\rho$. Further, $U$ is defined by $U(x):=F^{\leftharpoonup}\!\left(1-\frac{1}{x}\right)$, where $F^{\leftharpoonup}$ denotes the left inverse of the distribution function $F$. In this setting the following asymptotic normality statements for the Hill estimator $\hat{\gamma}_k$ hold.

\begin{thm}[Theorem 3.2.5 in \cite{deHaan-book}]\label{T:asymp_Hill}
Let $X_1,\dots,X_n$ be i.i.d.\ random variables with distribution function $F\in \mathrm{DoA}(G_{\gamma})$ for $\gamma>0$. If (\ref{E:second order condition}) holds and $k$ is an intermediate sequence, i.e.\ $k\rightarrow\infty$ and $k/n\rightarrow0$ as $n\rightarrow\infty$, then
\begin{align*}
\sqrt{k}(\hat{\gamma}_{k} - \gamma) \overset{\mathcal{D}}{\longrightarrow}\  \mathcal{N}\left(\frac{\lambda}{(1-\rho)}, \gamma^2\right),
\end{align*}
with $\lambda:= \underset{k\rightarrow\infty}{\lim} \sqrt{k}A(n/k)$. 
\end{thm}

\begin{thm}\label{T:1over_Hill}
Under the conditions of Theorem \ref{T:asymp_Hill}, it holds that
\begin{align*}
\sqrt{k} \left(  \frac{1}{\hat{\gamma}_k}-\frac{1}{\gamma} \right) \overset{\mathcal{D}}{\longrightarrow}\ \mathcal{N}\left( \frac{-\lambda }{(1-\rho) \gamma^2}, \frac{1}{\gamma^2} \right).
\end{align*}
\end{thm}
\begin{proof}
Applying the delta method to Thm.\ \ref{T:asymp_Hill}.
\end{proof}

Following the reasoning in \cite{deHaan-book}, page 78, the minimizing point of the AMSE can be found explicitly if considering $A(t)=ct^{\rho}$ with $c\neq0$. In this special case the minimizing sample fraction can be expressed as
\begin{equation}\label{E:k_opt}
 k_{\mathrm{opt}}= \left[  \left( \frac{\gamma^2 (1-\rho)^2}{-2\rho c^2} \right)^{1/(1-2\rho)} n^{-2\rho/(1-2\rho)} \right].
\end{equation}
Under the same assumption we can calculate the minimizing point $k_{\IHS}$ of the asymptotic expectations of $\IHS$ and $\IHS^{-}$. Let $\mathbb{AE}$ denote the asymptotic expectation referring to the expectation of the limiting distribution in Thm. \ref{T:1over_Hill}. Then
\begin{align*}
k_{\IHS} &:= \argmin_{k}\ \mathbb{AE}[\IHS]=\argmin_{k}\left\{ \frac{2}{\gamma k} +\frac{A(n/k)}{2\gamma^2(1-\rho) }\cdot \frac{k-4}{k} \right\} \\
&\approx \argmin_{k}\left\{ \frac{2}{\gamma k} +\frac{A(n/k)}{2\gamma^2(1-\rho) } \right\} = \left[ \left( \frac{4\gamma(1-\rho)}{-\rho c}\right)^{1/(1-\rho)} n^{-\rho/(1-\rho)}  \right].
\end{align*}
It is easy to check that the same formula holds for $\IHS^{-}$ if $c$ is replaced by its absolute value. Further note that by Lemma \ref{L:SKK2_range_k} it is sufficient to consider intermediate sequences when determining the minimizing sequence.
Comparing $k_{\mathrm{opt}}$ and $k_{\IHS}$ for a fixed $\rho>-\infty$ we obtain that
\begin{align}\label{E:k_SKKproportion}
\frac{k_{\IHS}}{k_{\mathrm{opt}}} \approx \left(  \frac{-\rho}{32}\cdot k_{\IHS}\right)^{-1/(1-2\rho)}\approx d\cdot n^{\frac{\rho}{(1-2\rho)(1-\rho)}} \longrightarrow 0,
\end{align}
as $n\rightarrow\infty$ and for a constant $d$ depending on $\rho$, $\gamma$ and $c$.
This supports what equation (\ref{E:SKKvsISE}) already suggested: minimizing $\IHS$ gives asymptotically a smaller $k$ than $k_{\mathrm{opt}}$. Thus, $k_{\IHS}$ asymptotically performs suboptimally for the Hill estimator but still leads to a consistent sequence of estimates. 
For finite samples the ratio crucially depends on $\rho$, and $k_{\IHS}$ can be even larger than $k_{\mathrm{opt}}$, as illustrated in Figure \ref{F:kSKK_over_Kopt}. The graphic presents the quotient of the two sample fractions as a function in the second order parameter $\rho$ for different samples sizes. The parameters $c$ and $\gamma$ are fixed to 1, as they have a weaker impact on the proportion. It also holds that $k_{\IHS}/k_{\mathrm{opt}}\rightarrow1$, as $\rho\rightarrow -\infty$, since both sample fractions converge to $n$ in this case.\\

\begin{figure}
\centering
\begin{minipage}[c]{0.5\textwidth}
	\includegraphics[width=\textwidth, height=0.28\textheight]{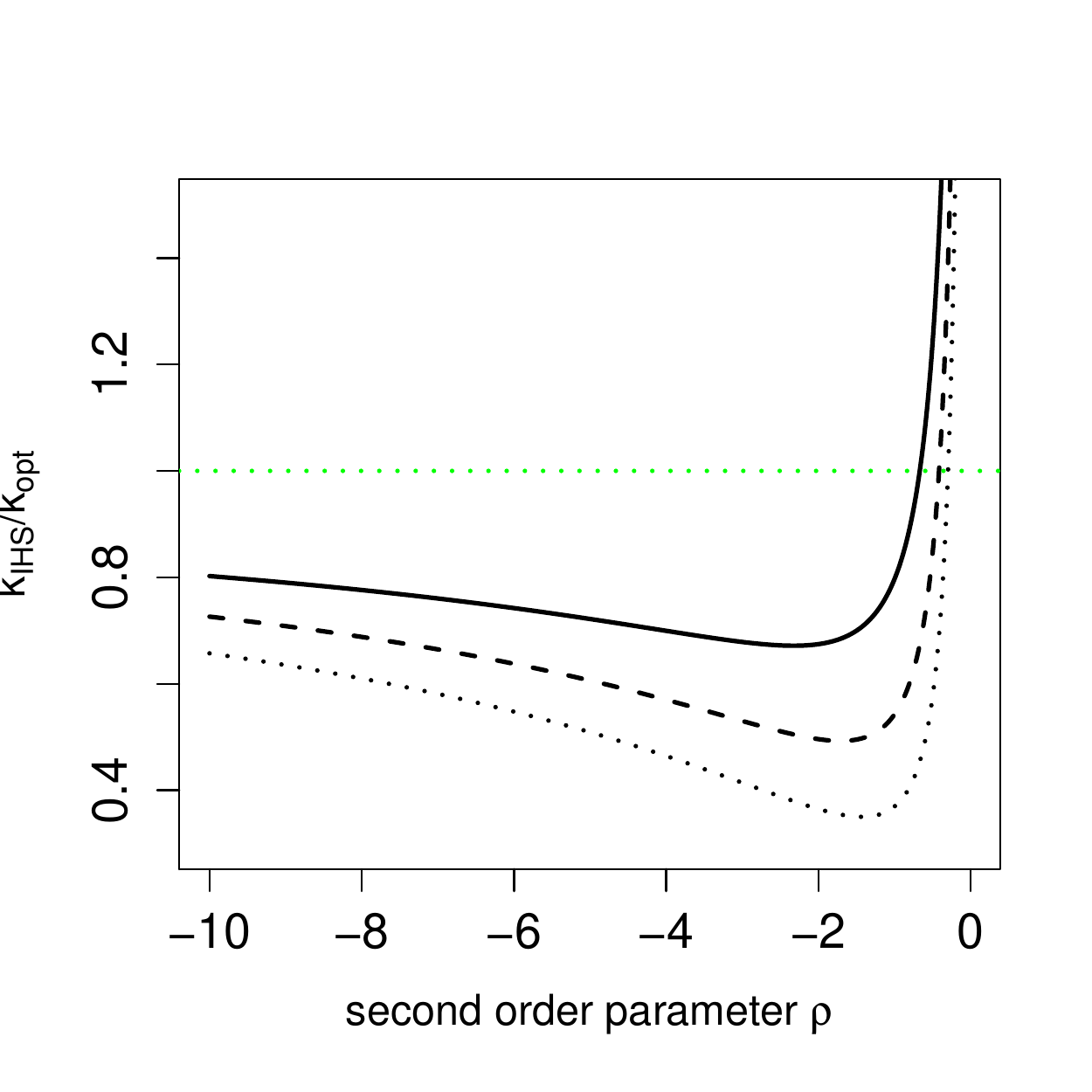}
  \end{minipage}\hfill
  \begin{minipage}[c]{0.48\textwidth}
    \caption{The approximation of the proportion $k_{\IHS}/k_{\mathrm{opt}}$ in (\ref{E:k_SKKproportion}) is plotted as a function in $\rho$ for $\gamma=c=1$ and $n=500$ (solid), $n=5000$ (dashed) and $n=50000$ (dotted).}\label{F:kSKK_over_Kopt}
  \end{minipage}
\end{figure}
%

Although $k_{\IHS}$ is of smaller order than $k_{\mathrm{opt}}$ asymptotically, the simulation study in Section \ref{S:Simulation} shows that $\hat{k}_{\IHS}$ works remarkably well when used for quantile estimation. We consider the following quantile estimator for the $(1-p)$-quantile,
\begin{equation}\label{E:q_hat}
 \hat{q}_k(p) = X_{(n-k,n)} \left(\frac{k}{np}  \right)^{\hat{\gamma}_k}.
\end{equation}
The sample fraction $k_{\mathrm{opt}}$ also minimizes the asymptotic relative MSE of $\hat{q}_k(p)$, see e.g.\ Theorem 4.3.8 in \cite{deHaan-book}. For finite samples however, the quantile estimator seems to benefit from $k_{\IHS}$. This has different reasons, two of which are illustrated by Figure \ref{F:Expectation_Plots}. On the left we see the empirical expectation of IHS, the empirical versions of the MSE of $\hat{\gamma}_k$ and the relative MSE of the quantile estimator,
\begin{equation}
\mathrm{MSEQ}:= \mathbb{E}\left[ \left( \frac{ \hat{q}_k(p)}{q(p)}-1   \right)^2 \right]/ \log\left(\frac{k}{np}\right),
\end{equation}
as used in Theorem 4.3.8 in \cite{deHaan-book}.
We observe that $k_{\IHS}$ (blue dot) is indeed smaller than $k_{\mathrm{opt}}$ (black) but so is the minimizer of MSEQ (pink) as well.\\
On the right we see a plot of the empirical $\mathbb{E}[\IHS]$ and MSE of $\hat{\gamma}_k$ for Loggamma distributed samples of size 5000. This graphic highlights the similarities between MSE and IHS for the boundary case $\rho=0$.\\
These observations indicate why $\hat{k}_{\IHS}$ outperforms other methods that try to minimize the MSE of the Hill estimator when adaptively estimating $q(p)$ by (\ref{E:q_hat}) on most of our exemplary distributions and sample sizes $n=500$ and $n=5000$, see Section \ref{S:Simulation}.

\begin{figure}
\centering
\includegraphics[width=0.48\textwidth, height=0.28\textheight]{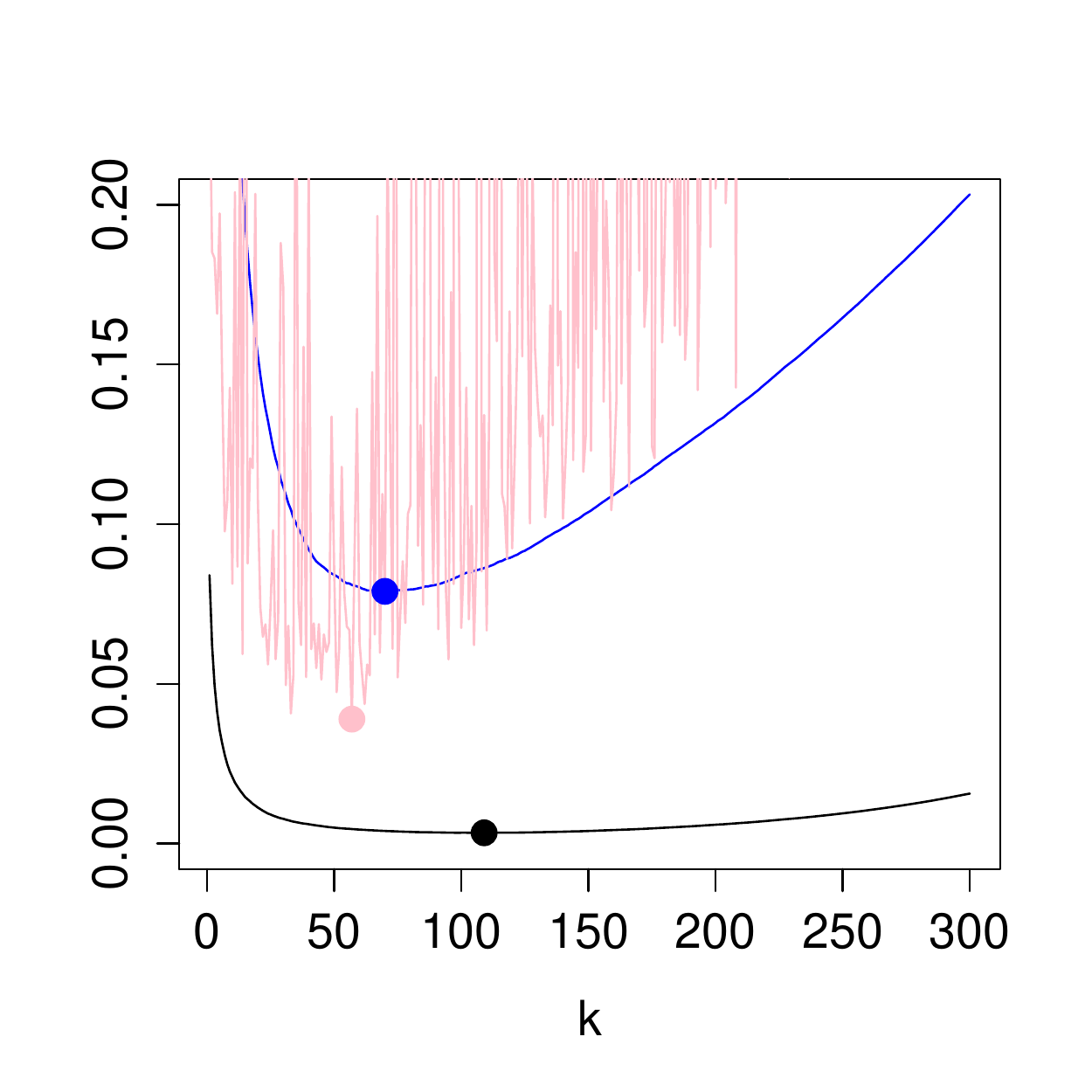}\quad
\includegraphics[width=0.48\textwidth, height=0.28\textheight]{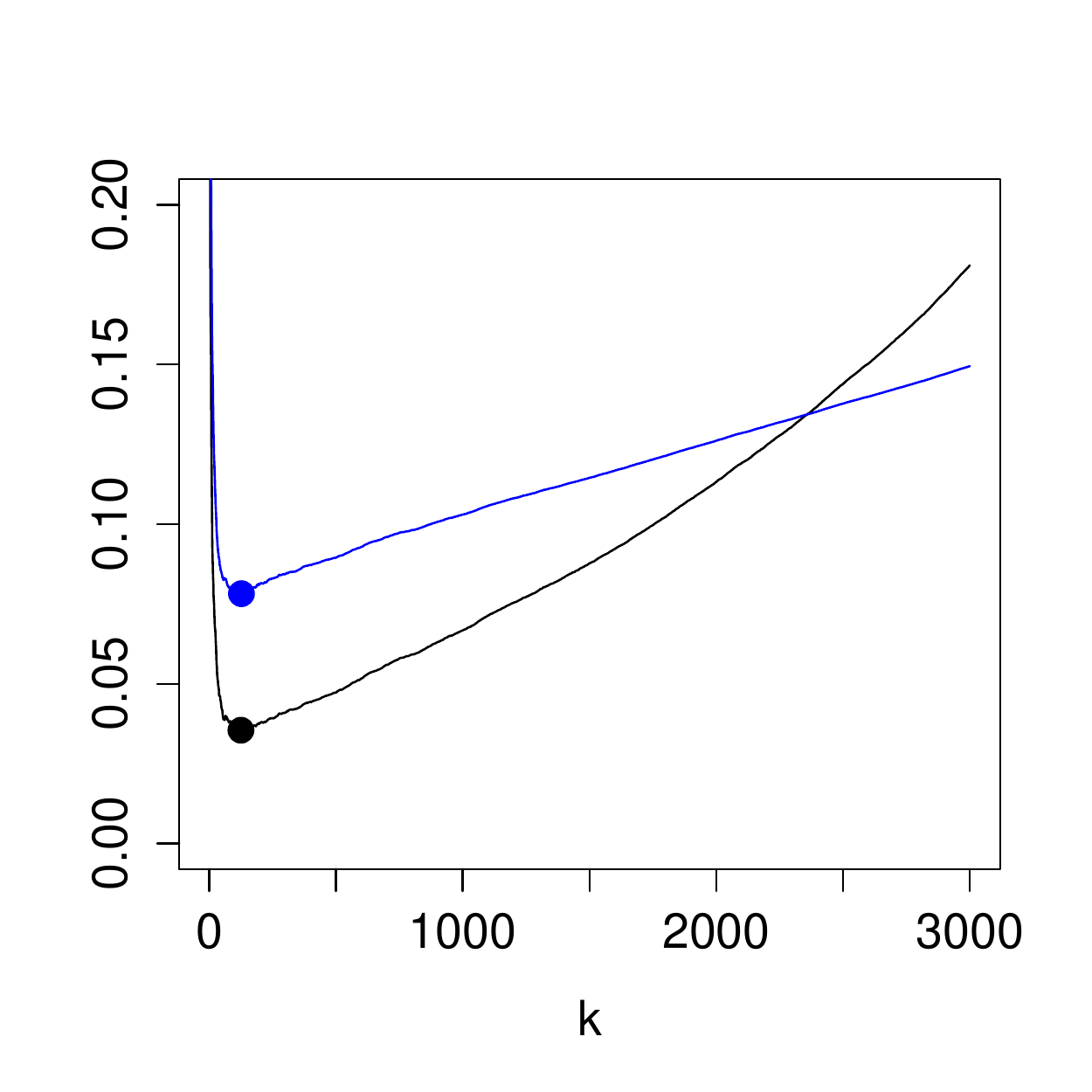}
\caption{Empirical expectations of IHS (blue), MSE (black) and MSEQ (pink). The left plot is based on 10,000 samples from a Fr\'echet(2) distribution of size 500. The graphic on the right is based on 500 samples of size 5000 from a Loggamma distribution.}\label{F:Expectation_Plots}
\end{figure}

\section{SAMSEE - The smooth AMSE estimator}\label{S:estimateMSE}

In this section we illustrate a way to smoothly estimate the AMSE of the Hill estimator. Via minimizing this AMSE estimator, called SAMSEE, we obtain an estimator for $k_{\mathrm{opt}}$. By this means, we extend previous methods which also estimate $k_{\mathrm{opt}}$ by estimating the AMSE itself.
From Thm.\ \ref{T:asymp_Hill} it is easy to see that the AMSE, which is the asymptotic variance plus the asymptotic squared bias, equals
\begin{equation}\label{E:asymp_MSE}
\mathbb{AE}[(\hat{\gamma}_k-\gamma)^2]= \frac{\gamma^2}{k} + \frac{A(n/k)^2}{(1-\rho)^2}.
\end{equation}
Thus, to estimate the AMSE as a function in $k$ we employ two estimators, one for $\gamma$ and one for the bias term as a combination of $\rho$ and $A$. First we explain how we estimate $\gamma$ and then we define the bias estimator. This bias estimator has a quite smooth sample path in $k$, and it depends on the choice of a large sample fraction $K$, for which we afterwards provide a data-driven selection procedure.\\
Note that, for the moment, we assume that the second order parameter $\rho$ is equal to $-1$ to motivate the construction of the AMSE estimator. The idea of misspecifying $\rho$ to simplify estimation -- via avoiding the additional uncertainty through estimating $\rho$ or selecting an influential tuning parameter -- was already used, for example, by \cite{Gomes01}, \cite{DreesK} and \cite{Goegebeur}. It is also motivated by the simulations in Section \ref{S:rho not -1}.\\
For $\gamma$ we consider the generalized Jackknife estimator $\hat{\gamma}_k^{\mathrm{GJ}}$ introduced by \cite{Gomes01} as $\gamma_n^{\mathrm{G}_1}$. This estimator is defined by
\begin{equation}\label{E:GJ}
 M_{n,k}:= \frac{1}{k}\sum_{i=1}^k Y_{(i,k)}^2,\quad \hat{\gamma}_{\mathrm{V},k}:=\frac{M_{n,k}}{2\hat{\gamma}_k} , \text{ and} \quad   \hat{\gamma}_k^{\mathrm{GJ}}:= 2\hat{\gamma}_{\mathrm{V},k} - \hat{\gamma}_k,
\end{equation}
where $Y_{(i,k)}$ denotes the log-spacings as in equation (\ref{E:Hill}).
Note, that $\hat{\gamma}_{V,k}$ is the de Vries estimator introduced under this name in \cite{deHaan-Peng} and $\hat{\gamma}_k$ is the Hill estimator as above.
The generalized Jackknife estimator has a reduced bias compared to the Hill estimator and is even asymptotically unbiased if $\rho=-1$, see (2.11) in \cite{Gomes01}. This property is useful here, since the bias estimator $\bar{b}_{\mathrm{up},K,k}$ defined in the following performs optimally for $\rho=-1$ as well. Furthermore, the same large sample fraction $K$ can be used for $\hat{\gamma}_K^{\mathrm{GJ}}$ and $\bar{b}_{\mathrm{up},K,k}$.\\
To construct this bias estimator, we study the following averages of Hill estimators,
\begin{equation}
\bar{\gamma}_k := \frac{1}{k}\sum_{i=1}^k  \hat{\gamma}_i \quad \text{and} \quad
\bar{\gamma}_{\mathrm{up},K,k} := \frac{1}{K-k+1}\sum_{i=k}^K \hat{\gamma}_i, \nonumber
\end{equation}
where $k<K$.
Plotting these averages illustrates how they smoothly frame the sample path of the Hill estimator. Especially the upper mean $\bar{\gamma}_{\mathrm{up},K,k}$ seems to contain a lot of structural information about the underlying asymptotic bias of the Hill estimator when choosing the upper bound $K$ appropriately, see Figure \ref{Fig:Hill_up_low}.
\begin{figure}
\centering
\begin{minipage}[c]{0.5\textwidth}
\includegraphics[width=0.98\textwidth, height=0.31\textheight]{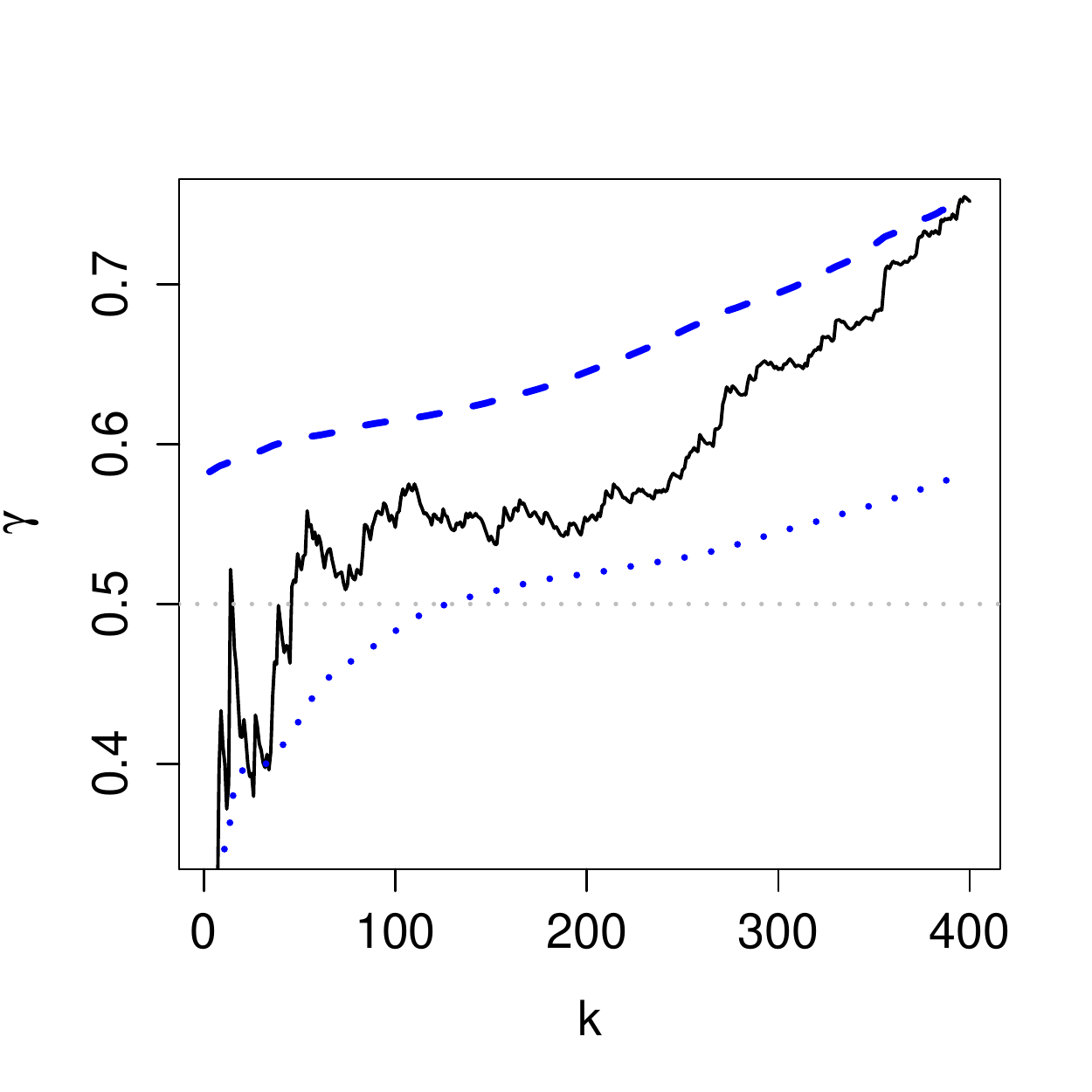}
\end{minipage}
\begin{minipage}[c]{0.48\textwidth}
\caption{The plot shows the Hill estimator (black), $\bar{\gamma}_k$ (blue dotted) and $\bar{\gamma}_{\mathrm{up},K,k}$ (blue dashed) with $K=400$ for a Fr\'echet(2) sample of size $n=500$ with true extreme value index $1/2$. }\label{Fig:Hill_up_low}
\end{minipage}
\end{figure}
This similarity between the upper mean and the bias of the Hill estimator inspires the definition
\begin{equation}\label{E:bias_est}
\bar{b}_{\mathrm{up},K,k}:= \bar{\gamma}_{\mathrm{up},K,k}-\bar{\gamma}_K.
\end{equation}
The estimator $\bar{b}_{\mathrm{up},K,k}$ is indeed a sensible estimator for a bias function, since
\begin{equation} \label{E:bias_expectation}
\mathbb{AE}[\bar{b}_{\mathrm{up},K,k}]=\frac{-\rho A(n/k)}{(1-\rho)^2} = \frac{1}{2} \frac{A(n/k)}{(1-\rho)} 
\end{equation}
follows for $\rho=-1$ from Theorem \ref{T:Bias}.\\
\cite{Danielsson} use $(\hat{\gamma}_{\mathrm{V},k}-\hat{\gamma}_k)$ to access the bias of $\hat{\gamma}_k$ and apply a double bootstrap procedure to stabilize this highly varying estimate. We use the difference of two estimators for $\gamma$ as well, but now consider averaging to smooth the bias estimate. The idea to average the Hill estimator in order to smooth the Hill plot and decrease the variance is also studied in \cite{Resnick97}.\\
It remains to choose an appropriate $K$ in order to complete SAMSEE and to estimate the optimal sample fraction $k_{\mathrm{opt}}$. We need $K$ to be large enough to allow for minimization over all relevant $k$ and small enough to be an intermediate sequence itself (see Theorem \ref{T:Bias} for this condition). To find such a $K$ we use the following relation between the estimators,
\begin{align}\label{E:asympEAprox}
\mathbb{AE}[\hat{\gamma}_k] = \mathbb{AE}[\hat{\gamma}_{\mathrm{V},k}+ \bar{b}_{\mathrm{up},K,k}].
\end{align}
This provides us with a relatively stable function in $k$, $\hat{\gamma}_{\mathrm{V},k}+ \bar{b}_{up,K,k}$, that has the same asymptotic expectation as the highly non-smooth Hill estimator. We want to find an intermediate sequence $K$ for which (\ref{E:asympEAprox}) holds and thus define
\begin{equation}\label{E:E2}
E^2(K):=\frac{1}{K}\sum_{k=1}^{K} \left( \hat{\gamma}_{\mathrm{V},k}+ \bar{b}_{\mathrm{up},K,k} - \hat{\gamma}_k \right)^2
\end{equation}
to measure the deviation from approximation (\ref{E:asympEAprox}) uniformly over all $k\leq K$. Based on this, we suggest to choose
\begin{align}\label{E:Kstar}
K^{*} :=\underset{K}{\argmin}\left\{\sum_{L=K-2}^{K+2} \Big(E^2(K)- E^2(L) \Big)^2 \right\}.
\end{align}
In this way we select a $K^{*}$ where the asymptotic approximation (\ref{E:asympEAprox}) is most stable, since we minimize the local variation of $E^2(K)$. Simulations suggest that this criterion is not sensitive to slightly increasing the region of stability $\{K-h,\dots,K+h\}$ from $h=2$ to $h=5$ or $10$ depending on the sample size.\\

Now we finally combine the previously described estimators to approach the AMSE in (\ref{E:asymp_MSE}) under the assumption that $\rho=-1$. With $K^{*}$ in (\ref{E:Kstar}) and the property of $\bar{b}_{\mathrm{up},K,k}$ in (\ref{E:bias_expectation}), we obtain an estimator for the AMSE of the Hill estimator and for $k_{\mathrm{opt}}$ by
\begin{align}
\Samsee(k) &:= \frac{(\hat{\gamma}_{K^{*}}^{\mathrm{GJ}})^2}{k} + 4\bar{b}_{\mathrm{up},K^{*},k}^2,\label{E:M_hat} \\
\hat{k}_{\Samsee}&:= \underset{1<k< K^{*}}{\mathrm{argmin}}\ \Samsee(k).\nonumber
\end{align}
Figure \ref{F:AMSE_Frechet2} illustrates how such a smooth estimate of the AMSE can look like. On the left, SAMSEE is displayed for a Fr\'echet sample with parameters $\gamma=1/2$ and $\rho=-1$. On the right, the Hill plot of the same sample is presented for all  $k\leq K^{*}=388$. \\

This smooth estimate of the AMSE can be useful beyond the context of threshold selection. For extreme value mixture models or Bayesian threshold selection approaches, SAMSEE could be used to construct a transition function between bulk and tail distribution or an empirical prior for the threshold, respectively, see \citet{review} for a review on mixture models.
\begin{figure}
\centering
\includegraphics[width=0.48\textwidth, height=0.28\textheight]{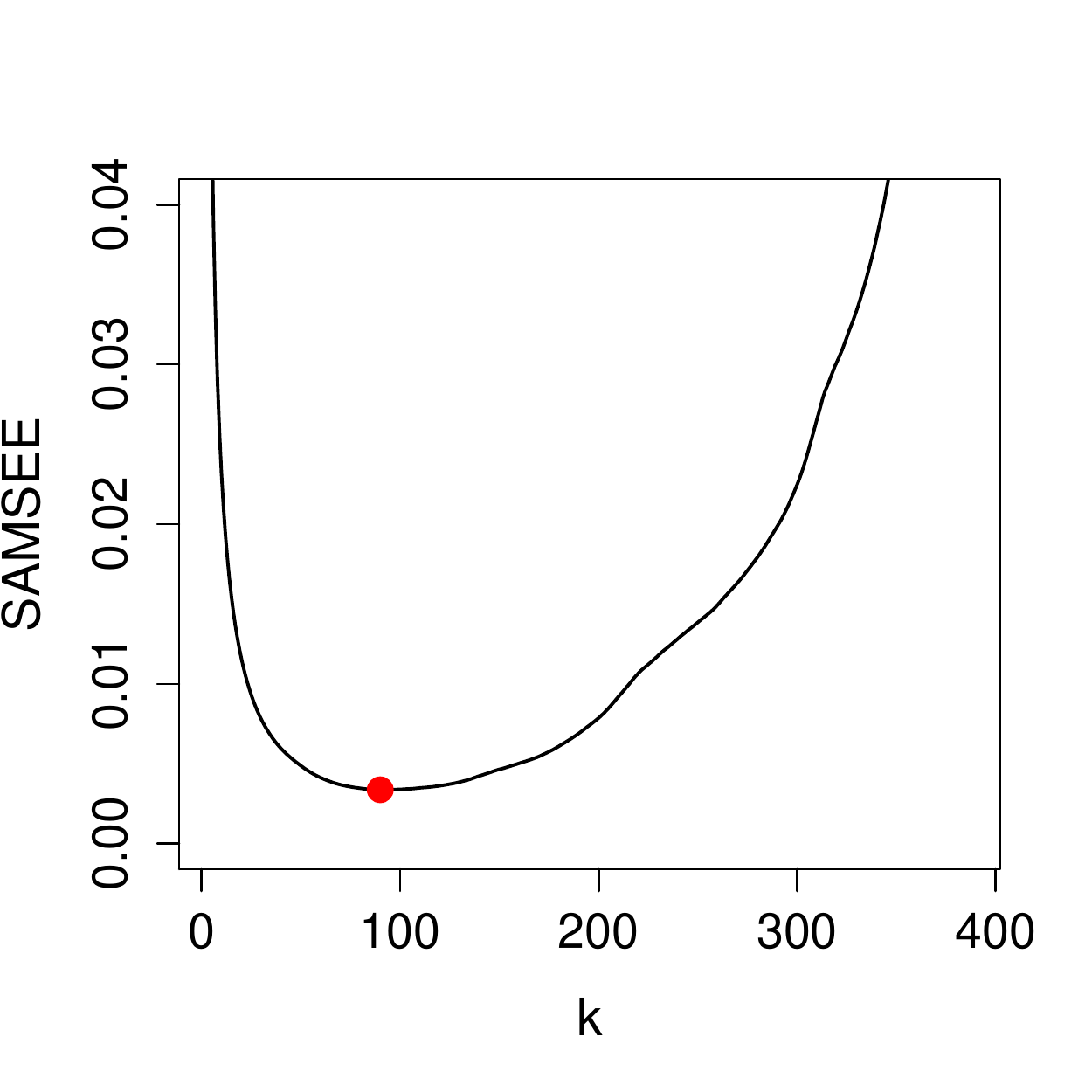}
\includegraphics[width=0.48\textwidth, height=0.28\textheight]{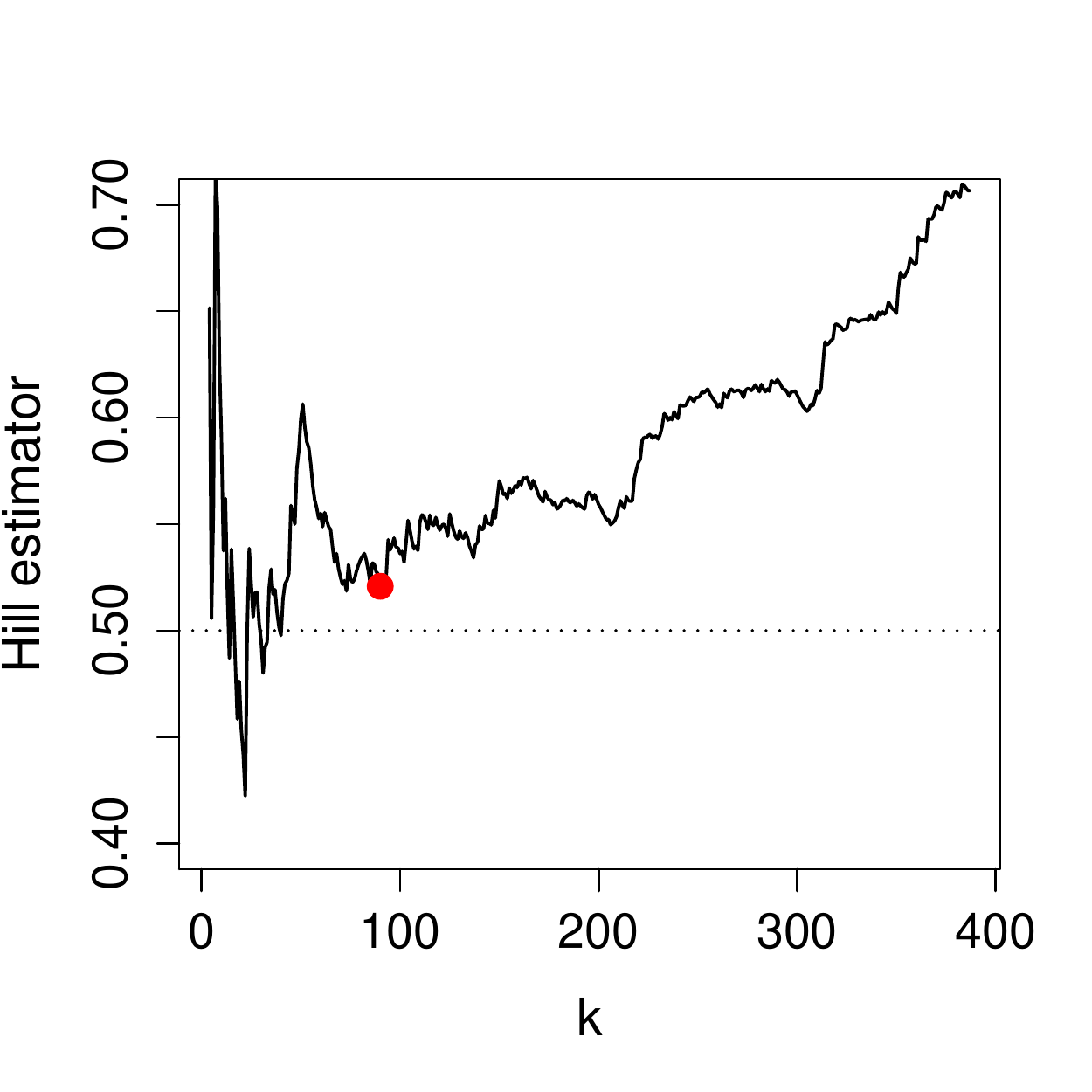}
\caption{SAMSEE with $K^{*}=388$ on the left next to the Hill plot for the same Fr\'echet(2) random sample of size $n=500$ for $k\leq K^{*}$. The red dot indicates the selected sample fraction $\hat{k}_{\Samsee}$ and the adaptive Hill estimate $\hat{\gamma}_{\hat{k}_{\Samsee}}$ in the right plot.}\label{F:AMSE_Frechet2}
\end{figure}

\subsection{SAMSEE if $\rho\neq -1$}\label{S:rho not -1}

We next want to analyse SAMSEE in the broader context of an unknown second order parameter $\rho$. The first thing to note is that the generalized Jackknife estimator is no longer unbiased in this situation. Secondly, the behaviour of our bias estimator changes, as it is described in the following Theorem.

\begin{thm}\label{T:Bias}
Under the conditions of Theorem \ref{T:asymp_Hill} and for $k/K\rightarrow c$ with $0<c<1$ as $n\rightarrow\infty$, it holds for $\bar{b}_{\mathrm{up},K,k}$ in (\ref{E:bias_est}) that
\begin{align*}
\sqrt{k}\cdot\bar{b}_{\mathrm{up},K,k} \Darrow \mathcal{N}\left( \frac{-\rho\lambda}{(1-\rho)^2}\delta_{\rho}(c),\ \gamma^2 \nu(c)  \right),
\end{align*}
where $\delta_{\rho}(c)=(c^{\rho}-1)/(-\rho(c^{-1}-1))$ and $\nu(c)=2c^2/(1-c)^2\cdot(1-c+c\log(c))$ with $0\leq\nu(c)\leq 1$.
\end{thm}
\begin{proof}
The proof can be found at the end of Appendix \ref{S:proofs}.
\end{proof}

From Theorem \ref{T:Bias} follows that
\begin{equation}
\mathbb{AE}[\bar{b}_{\mathrm{up},K,k}] = \frac{-\rho A(n/k)}{(1-\rho)^2} \cdot \delta_{\rho}(k/K). \nonumber
\end{equation} 
For $\rho=-1$ the function $\delta_{-1}(c)$ is equal to 1. If $\rho\neq -1$, we can observe that $\delta$ bends our bias estimator and it will therefore increase slightly too fast or too slow.
We can still apply SAMSEE in this situation and select $K^{*}$ from (\ref{E:Kstar}). However, approximation (\ref{E:asympEAprox}) does not hold anymore and instead the following holds,
\begin{align}
\mathbb{AE}[\hat{\gamma}_k - \hat{\gamma}_{\mathrm{V},k}]- \mathbb{AE}[\bar{b}_{\mathrm{up},K,k}] = \frac{-\rho A(n/k)}{(1-\rho)^2}\left( 1-\delta_{\rho}(k/K)  \right) .\label{E:asympEApox_rho}
\end{align}
The absolute value of the error described by (\ref{E:asympEApox_rho}) is high if $\delta_{\rho}$ strongly differs from $1$ and the bias term $A$ is large. If $\rho\neq-1$, $\delta_{\rho}$ indeed deviates from $1$ and we minimize the error by minimizing the bias. This is why applying (\ref{E:Kstar}) in this case leads to a small $K^{*}$. On the other hand, if $\delta_{\rho}=1$, the approximation stays valid for an increasing bias and $K^{*}$ will typically be larger.\\

An alternative to fixing $\rho=-1$ is to incorporate a consistent estimator $\hat{\rho}$ of the second order parameter. This can be done via
\begin{align*}
&\qquad  K_{\hat{\rho}}^{*}  :=\underset{K}{\argmin}\left\{\sum_{L=K-2}^{K+2} \Big(E_{\hat{\rho}}^2(K)- E_{\hat{\rho}}^2(L) \Big)^2 \right\},\\
&\text{where }E_{\hat{\rho}}^2(K)  :=\frac{1}{K}\sum_{k=1}^{K} \left( \hat{\gamma}_{\mathrm{V},k}+ \bar{b}_{\mathrm{up}, K,k}/\delta_{\hat{\rho}}(k/K) - \hat{\gamma}_k \right)^2.
\end{align*}
and 
\begin{align}
\Samsee_{\hat{\rho}}(k)&:= \frac{(\hat{\gamma}_{ K_{\hat{\rho}}^{*}}^{\mathrm{GJ}})^2}{k} + \left((1-\hat{\rho}) \frac{ K_{\hat{\rho}}^{*}/k-1}{(k/ K_{\hat{\rho}}^{*})^{\hat{\rho}}-1}\cdot \bar{b}_{\mathrm{up}, K_{\hat{\rho}}^{*},k}\right)^2,\label{E:M_hat_rho} \\
\hat{k}_{\hat{\rho},\Samsee}&:= \underset{1<k< K^{*}}{\mathrm{argmin}}\ \Samsee_{\hat{\rho}}(k).\nonumber
\end{align}
In this way we can construct an estimator for $k_{\mathrm{opt}}$ in the general setting of Pareto-type distributions.\\
In Table \ref{Tab:rho}, we present the results of a simulation study indicating for which distributions it is beneficial to use $\hat{\rho}$ instead of $\rho=-1$. We estimate $\rho$ using the estimator $\hat{\rho}^{(1)}$ suggested in Theorem 1 in \cite{DreesK}. The results indicate that, in general, it is sensible to fix $\rho=-1$ in SAMSEE, since only for the Cauchy distribution using $\hat{\rho}$ performs slightly better regarding bias and RMSE. This confirms the observations already made by others \citep{Gomes01, DreesK, Goegebeur}, that it is often recommendable to select $\rho=-1$ instead of allowing for further variability by including an additional estimator.

\begin{table}[ht]
\centering
\begin{tabular}{r|cc|ccc}
\hline
\multicolumn{3}{c}{ }& \multicolumn{3}{|c}{$\mathbb{E}[\hat{\gamma}_{\hat{k}}]$ (RMSE)}\\
\multicolumn{1}{c}{ }& $\gamma$ & $\rho$ & true $\rho$ & $\rho=-1$ & $\hat{\rho}$ \\ 
  \hline
Student-t(6)  & 0.17 &-1/3 & 0.21 (0.09)& 0.26 (0.12) & 0.28 (0.14) \\ 
  Fr\'echet(2) & 0.50 & -1 & 0.51 (0.07) & 0.51 (0.07) & 0.51 (0.08) \\ 
  Cauchy & 1.00 & -2 & 1.01 (0.13) & 0.97 (0.17) & 0.99 (0.16) \\ 
  Burr(2,1) & 2.00 & -1& 2.05 (0.34) & 2.05 (0.34) & 2.03 (0.40) \\ 
   \hline
\end{tabular}
\caption{The averages of adaptive $\gamma$ estimates and their root mean square error (RMSE) in brackets are presented for thresholds $\hat{k}$ that are selected using SAMSEE or $\Samsee_{\hat{\rho}}$ with the true $\rho$, $\rho=-1$ or $\hat{\rho}=\hat{\rho}^{(1)}$.}\label{Tab:rho}
\end{table}

\section{Simulation study}\label{S:Simulation}

In the following we numerically analyse the performance of eight threshold selection methods on heavy-tailed distributions with very different tail behaviour. The simulation study is based on the following distributions:
\begin{itemize}
\item the Student-t distribution with 6 degrees of freedom, which corresponds to $\gamma=1/6$ and $\rho=-1/3$,
\item the Fr\'echet distribution with parameter $\alpha=2$ and distribution function $F(x)=\exp(-x^{-\alpha})$ for $x>0$, which implies $\gamma=1/2$ and $\rho=-1$,
\item the standard Cauchy distribution leading to a tail behaviour with $\gamma=1$ and $\rho=-2$,
\item the Loggamma distribution with $\gamma=1$ and $\rho=0$ and density function 
\begin{equation} f(x)=\log(x) x^{-2}\ \mathds{1}_{[1,\infty)}(x),\nonumber
\end{equation}
\item the Burr distribution with a parametrisation such that $\gamma=2$, $\rho=-1$ and distribution function  \begin{equation} F(x)=1-(1+\sqrt{x})^{-1},\ \text{for } x>0,\nonumber
\end{equation}
\item a logarithmically perturbed Pareto distribution of the random variable $g(U)$ with $\gamma=1$ and $\rho=-1$, where $U\sim \mathrm{Unif}(0,1)$ and $g(x)=x^{-1}/\log(x^{-1})$. This distribution is denoted as negBias due to its negative bias in the Hill estimator.
\end{itemize} 

On these distributions we evaluate the methods by their root mean square error (RMSE) when adaptively estimating $\gamma$ with the Hill estimator relative to the RMSE obtained using $k_{\mathrm{opt}}$, 
\begin{equation}
\mathrm{EFF}_{\gamma}(\hat{k}):=\sqrt{\frac{\mathbb{E}_n[(\hat{\gamma}_{\hat{k}}-\gamma)^2]}{\mathbb{E}_n[(\hat{\gamma}_{k_{\mathrm{opt}}}-\gamma)^2]}}, \nonumber
\end{equation}
where $\mathbb{E}_n$ denotes the empirical expectation. These efficiency quotients are also used by, e.g., \cite{Guillou}, \cite{Gomes01} and \cite{DreesK}. The smaller the quotient the better the threshold selection procedure performs compared to the asymptotically optimal sample fraction $k_{\mathrm{opt}}$. 
Furthermore, we study the efficiency in quantile estimation with the estimator defined in (\ref{E:q_hat}) for $p=0.001$,
\begin{equation}
\mathrm{EFF}_{q}(\hat{k}):=\sqrt{ \frac{\mathbb{E}_n[ (\hat{q}_{\hat{k}}-q)^2 ]}{\mathbb{E}_n[ (\hat{q}_{k_{\mathrm{opt}}}-q)^2]}}. \nonumber
\end{equation}
Since we do not know the true minimizer $k_{\mathrm{opt}}$ of the AMSE, we utilize an empirical version suggested by \citet{Gomes01}. Following their approach we approximate $k_{\mathrm{opt}}$ by the mean of 20 independent replicates of $\bar{k}_{\mathrm{opt}}$, which is the minimizer of the empirical MSE based on $1000$ samples, i.e.\ 
$
\bar{k}_{\mathrm{opt}}= \underset{k}{\mathrm{argmin}} \ \mathbb{E}_{n=1000} [ (\hat{\gamma}_k-\gamma)^2 ]. \nonumber
$\\

We compare these efficiency values for eight different threshold selection methods. Most of the considered approaches are constructed for adaptive estimation of $\gamma$ applying the Hill estimator. This includes one procedure that looks for a stable region among the Hill estimates, while the others aim to estimate $k_{\mathrm{opt}}$. The only exception is the IHS approach discussed in Section \ref{S:SKK}, which is motivated to minimize the deviation from the exponential approximation. We still evaluate the performance of this procedure in the same simulations, although it is not primarily tailored for the specific applications.
In total, the following methods are considered:
\begin{description}
\item[{\bf sIHS:}] IHS smoothed by using the \textit{eBsc} package, see Section \ref{S:SKK},
\item[{\bf SAM:}]  SAMSEE procedure with $\rho=-1$ as defined by (\ref{E:M_hat}) in Section \ref{S:estimateMSE},
\item[{\bf GH:}] method by \cite{Guillou} utilizing $c_{\mathrm{crit}}=1.25$ and $p=1$,
\item[{\bf DK:}] procedure by \cite{DreesK} with fixed $\rho=-1$,
\item[{\bf GO:}] approach by \cite{Goegebeur} defined in their equation (3.3) with fixed $\rho=-1$,
\item[{\bf DB:}] double bootstrap approach by \cite{Danielsson} with the choice $n_1=120$ if $n=500$ and $n_1=1000$ if $n=5000$,
\item[{\bf B:}] method by \cite{Beirlant02} with $\rho=-1$,
\item[{\bf RT:}] method by \cite{Reiss-Thomas} with $\beta=0$ as suggested by \cite{Neves-RT}. 
\end{description} 

\begin{table}[h]
\centering
\begin{tabular}{rccccccc|c}
  \hline
 $n=500$  & SAM & GH & DK & GO & DB & B & RT & sIHS \\ 
  \hline
Student-t(6) & 1.07 & 1.68 & 1.18 & 1.38 & 1.06 & \textcolor{blue}{1.04} & \textcolor{blue}{1.04} & 1.14 \\ 
  Fr\'echet(2) & 1.13 & 1.15 & \textcolor{blue}{1.08} & 1.12 & 1.60 & 1.49 & 2.00 & 1.41 \\ 
  Cauchy & 1.37 & 1.19 & 1.32 & \textcolor{blue}{1.16} & 2.14 & 1.85 & 2.11 & 1.47 \\ 
  Loggamma & 0.98 & 1.06 & 1.27 & 1.11 & 1.12 & 1.04 & 1.32 & \textcolor{blue}{0.78} \\ 
  Burr(2,1) & \textcolor{blue}{1.11} & 1.22 & 1.47 & 1.13 & 1.68 & 1.42 & 1.82 & 1.14 \\ 
  negBias & \textcolor{blue}{1.06} & 1.13 & 1.56 & 1.13 & 1.07 & 1.22 & 1.89 & 2.27 \\ 
   \hline
\end{tabular}
\vspace*{0.2cm}

\begin{tabular}{rccccccc|c}
  \hline
$n=5000$ & SAM & GH & DK & GO & DB & B & RT & sIHS \\ 
  \hline
Student-t(6) & 1.20 & 1.58 & 1.31 & 1.39 & 1.35 & \textcolor{blue}{1.03} & 1.26 & \textcolor{blue}{1.03} \\ 
  Fr\'echet(2) & 1.08 & 1.26 & \textcolor{blue}{1.07} & 1.21 & 1.66 & 1.29 & 2.40 & 2.43 \\ 
  Cauchy & 1.34 & 1.41 & \textcolor{blue}{1.08} & 1.17 & 2.00 & 1.68 & 2.78 & 3.03 \\ 
  Loggamma & 1.08 & 1.10 & 1.32 & 1.17 & 1.19 & 1.05 & 1.40 & \textcolor{blue}{0.79} \\ 
  Burr(2,1) & \textcolor{blue}{1.07} & 1.29 & 1.62 & 1.14 & 1.63 & 1.29 & 2.21 & 1.79 \\ 
 negBias & \textcolor{blue}{0.98} & 1.12 & 1.54 & 1.10 & 1.30 & 1.04 & 2.08 & 3.98 \\ 
   \hline
\end{tabular}
\caption{Efficiency values $\mathrm{EFF}_{\gamma}$ based on 2000 samples if $n=500$ and on 500 samples if $n=5000$. Lowest (best) efficiency values are highlighted in blue. }\label{Tab:EFFg}
\end{table}

\begin{table}[h]
\centering
\begin{tabular}{rccccccc|c}
  \hline
 $n=500$  & SAM & GH & DK & GO & DB & B & RT & sIHS \\ 
  \hline
Student-t(6) & 1.09 & 2.30 & 1.23 & 1.60 & \textcolor{blue}{ 1.01} & 1.04 & 1.16 & 1.10 \\ 
  Fr\'echet(2) & 0.96 & 1.07 & 1.01 & 1.06 & 1.07  & 1.16 & 1.42 & \textcolor{blue}{0.83} \\ 
  Cauchy & 0.89 & 1.04 & 1.03 & 0.95 & 0.86  & 1.40 & 1.59 & \textcolor{blue}{0.65} \\ 
  Loggamma & 0.84 & 0.95 & 2.10 & 1.02 & 0.88 & 1.06 & 1.55 & \textcolor{blue}{0.50} \\ 
  Burr(2,1) & 0.79 & 2.15 & 8.60 & 0.98 & 0.71 & 1.43 & 3.19 & \textcolor{blue}{0.41} \\ 
  negBias & 1.66 & 1.37 & 2.32 & 2.13 & \textcolor{blue}{ 0.80} & 1.98 & 3.75 & 8.05 \\ 
   \hline
\end{tabular}
\vspace*{0.2cm}

\begin{tabular}{rccccccc|c}
  \hline
 $n=5000$ & SAM & GH & DK & GO & DB & B & RT & sIHS \\ 
  \hline
Student-t(6) & 1.07 & 1.39 & 1.16 & 1.29 & 1.44 & 1.02 & 1.24 & \textcolor{blue}{0.94} \\ 
  Fr\'echet(2) & \textcolor{blue}{1.04} & 1.14 & 1.07 & 1.15 & 1.14 & 1.14 & 1.53 & 1.14 \\ 
  Cauchy & \textcolor{blue}{1.01} & 1.13 & 1.03 & 1.04 & 1.16 & 1.22 & 1.55 & 1.10 \\ 
  Loggamma & 0.94 & 1.00 & 1.39 & 1.12 & 1.11 & 0.97 & 1.55 & \textcolor{blue}{0.67} \\ 
  Burr(2,1) & 1.00 & 1.38 & 1.24 & 1.11 & 1.09 & 1.11 & 1.72 & \textcolor{blue}{0.71} \\ 
  negBias &1.00 & 1.11 & 0.87 & 1.21 &  \textcolor{blue}{0.87} & 1.21 & 1.23 & 2.16 \\ 
   \hline
\end{tabular}
\caption{Efficiency values  $\mathrm{EFF}_{q}$ for $p=0.001$ based on 2000 samples if $n=500$ and on 500 samples if $n=5000$. Lowest (best) efficiency values are highlighted in blue.}\label{Tab:EFFq}
\end{table}

When looking at the results for estimating $\gamma$ adaptively for $n=500$ and $n=5000$ in Table \ref{Tab:EFFg}, we observe a very diverse picture of methods performing best. Overall  we get the impression that SAMSEE together with the approach by \cite{Goegebeur} performs most stable over the variety of distributions. This is interesting, because those are the methods which depend least on tuning parameters. The performance of the approaches GH, DK and B is comparable, but we obtain from Table \ref{Tab:EFFg} that on average over all distributions the SAMSEE procedure is superior.\\
For estimating a high quantile SAMSEE also performs convincingly, see Table \ref{Tab:EFFq}, but additionally sIHS and the approach by \cite{Danielsson} show very good efficiency values. They are closely followed by B and GO. Looking at the average performance over all distributions, SAMSEE performs best again. However, if we exclude the negBias distribution, sIHS works superior on average.\\
In conclusion, we can see that SAMSEE performs very efficiently and comparable to $k_{\mathrm{opt}}$ over all exemplary distributions. It works especially well for estimating a high quantile. Only in the case of estimating $\gamma$ for the Cauchy distribution it performs worse than DK and GO, but still better than most other approaches. Recalling the results of the simulation on the influence of $\rho$ in Table \ref{Tab:rho} in Section \ref{S:rho not -1}, it is not very surprising that SAMSEE performs slightly weaker in this situation. There, the Cauchy distribution is the only example we considered that benefits from estimating $\rho$ instead of fixing it to $-1$.\\
From Table \ref{Tab:EFFq} we furthermore observe that sIHS is a strong choice when estimating high quantiles from small samples with $n$ up to $5000$. However, the performance when estimating $\gamma$ is quite variable and it seems that sIHS does not perform particularly well for distributions with a small second order parameter ($\rho\leq-1$). This behaviour is already discussed in Section \ref{S:SKK-theory} and highlighted in Figure \ref{F:kSKK_over_Kopt}: sIHS selects smaller $k$ than optimal for the Hill estimator, especially if $\rho$ is in the regime between $-1$ and $-8$.\\
The reason why some approaches perform worse on quantiles than they do on $\gamma$ is that the estimator $\hat{q}_k(p)$ defined in (\ref{E:q_hat}) depends on $\hat{\gamma}_k$ in the exponent and is thus very sensitive to overestimation in case of $\gamma>1$. Hence, when estimating a high quantile, an estimate $\hat{\gamma}_k$ that is too large will lead to an even stronger overestimation of the quantile. This is why a few outliers among the $\gamma$ estimates can already cause much higher $\mathrm{EFF}_q$ values.\\

\section{Application to varying extreme value index}\label{S:app}

In this section we analyse our new procedures in a financial application, where we study operational losses of a bank. We are, of course, particularly interested in the distributional properties of very high losses. It has been discussed before that it is reasonable to assume the distribution of such extreme losses being heavy-tailed \citep{Valerie-loss, Moscadelli} and to change with the financial market over time \citep{Hambuckers, Cope}. In this context, we want to estimate how the extreme value index changes depending on the univariate covariate time. For this task, we utilize the approaches presented in Sections \ref{S:SKK} and \ref{S:estimateMSE} for locally optimal selection of a threshold.\\
The observations of interest are operational losses from the Italian bank UniCredit from 2005 to 2014. In \cite{Hambuckers} the data is analysed in a regularized generalized Pareto regression approach including several firm-specific, macroeconomic and financial indicators as covariates. This approach describes the dependence of the GPD parameters on various covariates via parametric functions. \\
We consider an easier and more direct approach to study the temporal dependence of the extreme value index without taking into account possible interference by other covariates. Our aim is to estimate the time dependent extreme value index $\gamma(t)$ non-parametrically with a simple ad hoc estimator that extends the estimator from \cite{deHaan-Chen} by employing our threshold selection procedures sIHS and SAMSEE. We present the estimator in Section \ref{S:app-t}\label{S:app-t} and the results we obtain when applying this estimator to the dataset of operational losses in Section \ref{S:app-data}.

\subsection{Estimating a varying extreme value index}\label{S:app-t}

In \cite{deHaan-Chen}, the authors already discussed estimating a trend in the extreme value index non-parametrically. They consider $n$ independent random variables $X_i\sim F_{i/n}$, where $F_s\in \mathrm{DoA}(G_{\gamma(s)})$ for $s\in[0,1]$. To address this problem, they introduce the following estimator for $\gamma(s)$, which locally applies the Hill estimator and is based on a global sample fraction $k$,
\begin{equation}\label{E:G(s)_k}
\hat{\gamma}_k(s):= \frac{1}{2kh} \sum_{i\in I_{n}(s)} \left(  \log X_i - \log X_{(\lfloor 2nh\rfloor-\lfloor 2kh\rfloor, \lfloor 2nh\rfloor )}  \right)^{+},
\end{equation}
where $I_n(s)$ is the $h$-neighbourhood of $s$, i.e.\ $I_n(s):=\{ i: |i/n-s|\leq h\}$. This estimator depends on the choice of the bandwidth $h$ and the global sample fraction $k$, which is then rescaled to $2kh$ for the individual regions $I_n(s)$. A small bandwidth $h$ leads to very high variability in $\hat{\gamma}_k(s)$ and a large value of $h$ smooths out all interesting features. Thus, the choice of $h$ should balance these two effects. \\
We suggest a modification of their estimator, where we locally estimate an optimal threshold $\hat{k}(s)$, i.e.
\begin{equation}\label{E:G(s)_khat}
\hat{\gamma}_{\hat{k}(s)}(s):= \frac{1}{\hat{k}(s)} \sum_{i\in I_{n}(s)} \left(  \log X_i - \log X_{(\lfloor 2nh\rfloor-\hat{k}(s), \lfloor 2nh\rfloor )}  \right)^{+}.
\end{equation}
To compare these two approaches, we repeat the simulation presented in Figure 2 (i) in \cite{deHaan-Chen} on samples of size $n=5000$ with $X_i\sim \text{Fr\'echet}(1/\gamma(i/n))$ and $\gamma(s)=1+s$. Figure \ref{F:av_frechet_gs} illustrates the benefits of locally optimizing the threshold via SAMSEE from Section \ref{S:estimateMSE}, as it strongly tightens the empirical confidence interval around the average, which is obtained from the $2.5$\% and $97.5$\% quantiles among $1000$ estimates.\\

\begin{figure}
\centering
\begin{minipage}[c]{0.5\textwidth}
	\includegraphics[scale=0.5]{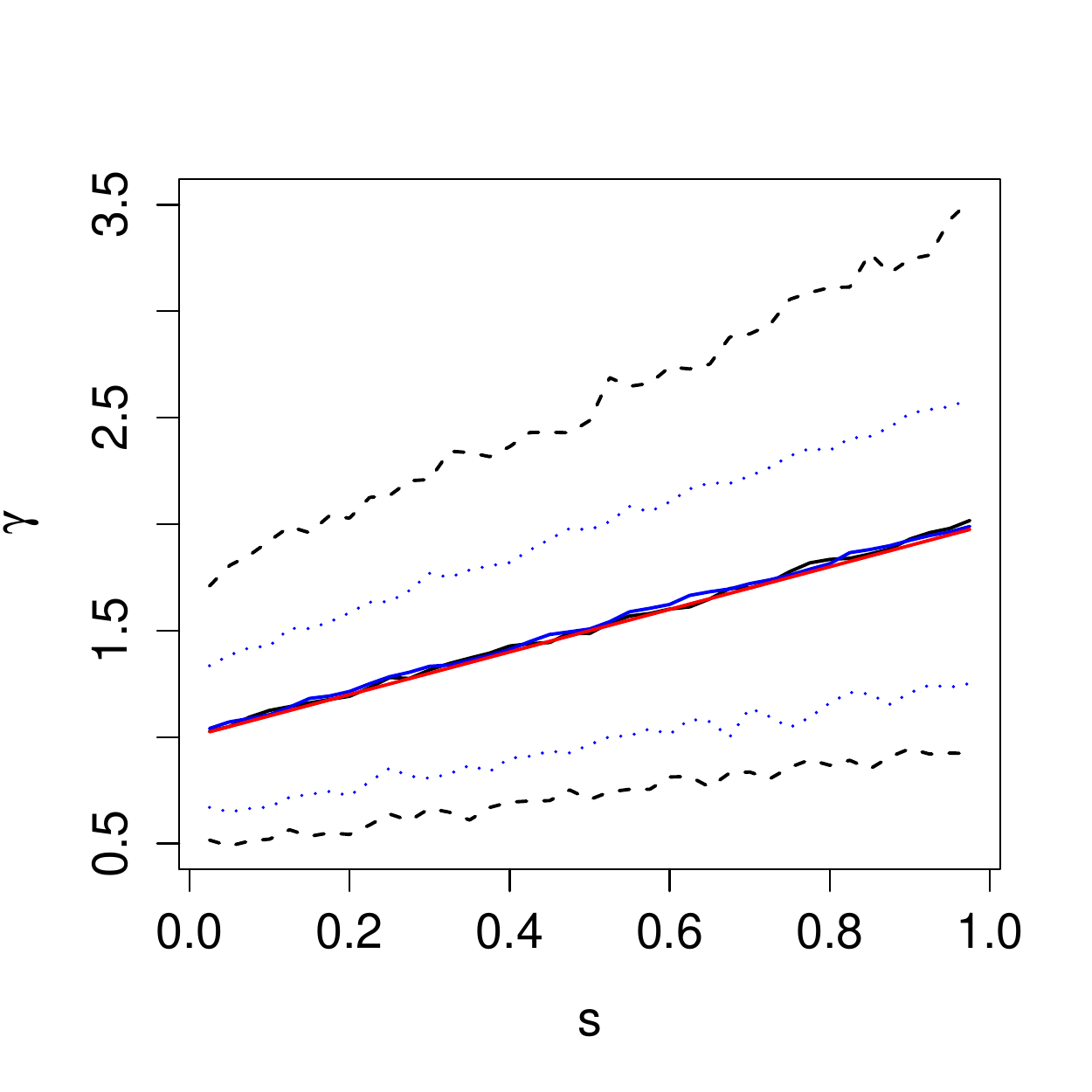}
  \end{minipage}\hfill
  \begin{minipage}[c]{0.48\textwidth}
\caption{The true extreme value index $\gamma(s)=1+s$ (red) next to the averaged estimators over 1000 samples. The estimator from \cite{deHaan-Chen} is in black and its empirical 95\% confidence interval is dashed. Our modified estimator employing SAMSEE is blue with dotted empirical confidence bounds.}\label{F:av_frechet_gs}
\end{minipage}
\end{figure}

\subsection{Functional extreme value index of operational losses}\label{S:app-data}

The operational losses in the dataset of UniCredit are grouped by the type of event that caused the specific loss. We consider the event type CPBP, which provides sufficient observations for our local estimation approach. The CPBP losses are caused by clients, products and business practices related to derivatives or other financial instruments.\\
First we want to test if the extreme value index is constant over time. Using the test T4 from \cite{Einmahl16}, we can reject the null hypotheses with a $p$-value that is virtually zero and thus are confident that the extreme value index of the losses is indeed varying over time.\\
We apply the new methodology from (\ref{E:G(s)_khat}) to these losses via estimating $\hat{k}$ with sIHS from Section \ref{S:SKK} and the SAMSEE approach from Section \ref{S:estimateMSE}. Figure \ref{F:CPBP_EFRAUD} shows the estimates we obtain for the event type CPBP. It is clearly visible that both procedures yield similar estimates for most time points and that the simple ad hoc estimators recover an increase of the severity of high losses during the financial and Euro crisis from 2008 to 2011. A similar overall trend in the extreme value index can also be identified in the estimates of \citet{Hambuckers} for CPBP. \\
For a more extensive discussion of the data and results of the more complex model including further covariates we refer to \cite{Hambuckers}.

\begin{figure}
\centering
\includegraphics[width=\textwidth]{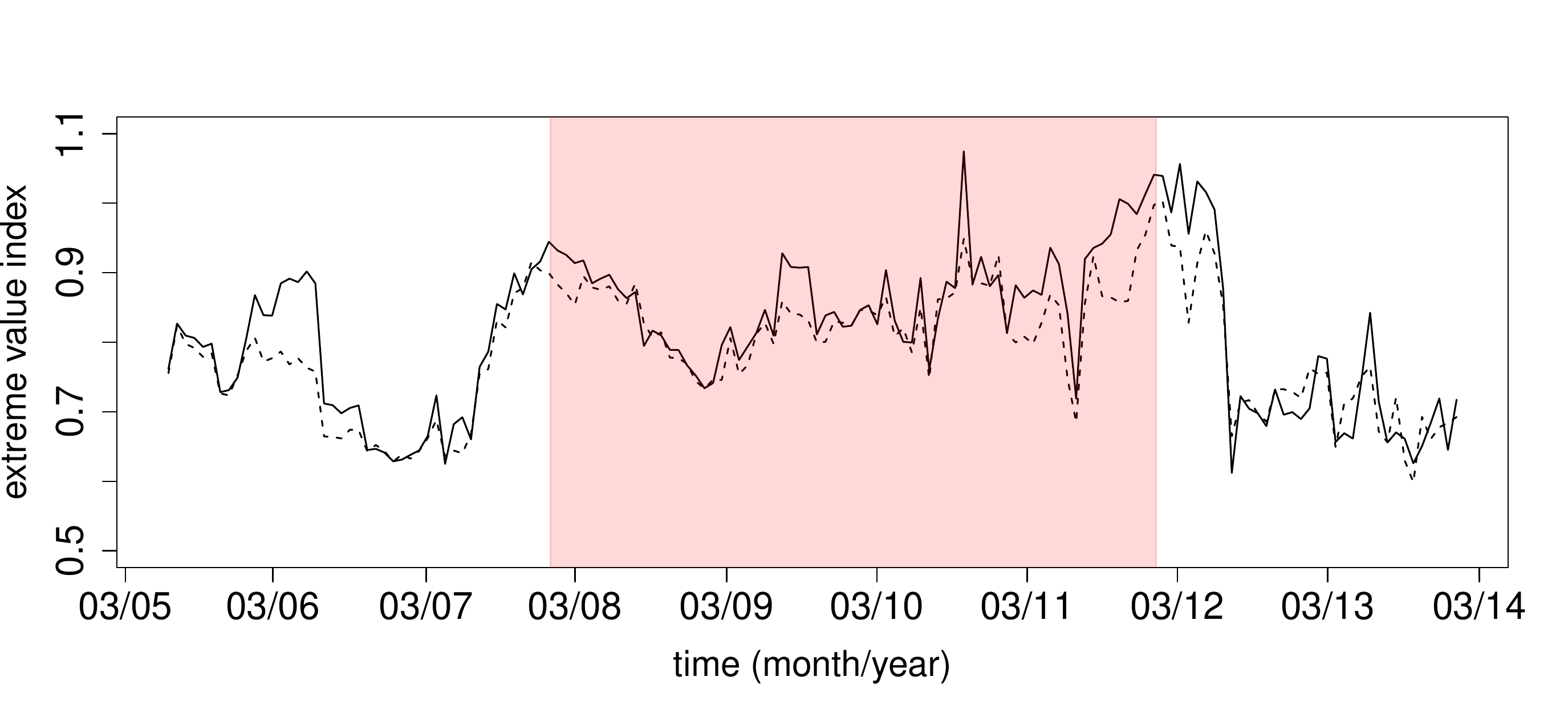}
\caption{The non-parametric estimate of the extreme value index of the operational losses of type CPBP is presented using $\hat{k}$ from SAMSEE (black) and sIHS (dashed) and bandwidth $h=0.05$. The red area indicates the time of the financial and Euro crisis.}\label{F:CPBP_EFRAUD}

\end{figure}

\appendix
\section{Theoretical results and proof of Theorem \ref{T:Bias}}\label{S:proofs}

\begin{lem}[Distribution of the Hill estimator]\label{L:Hill-dist}
Let $X_1,\dots,X_n\iid F$ for $F\in \mathrm{DoA}(G_{\gamma})$ with $\gamma>0$. Then the following distributional representation for the Hill estimator holds,
\begin{equation}
\hat{\gamma}_k = \frac{1}{k}\sum_{i=1}^k \log\left( \frac{X_{(n-i+1,n)}}{X_{(n-k,n)}} \right) \Deq G_k + b_{n,k}, \nonumber
\end{equation}
where $G_k\sim \Gamma(k,\gamma/k)$ and for $b_{n,k}$ it holds that
\begin{align*}
b_{n,k} \longrightarrow \begin{cases} 0, \text{ if } k/n\rightarrow0, \\ b_c, \text{ if } k/n\rightarrow c, \end{cases}  \text{ as } n\rightarrow\infty,
\end{align*}
for some $b_c\in\mathbb{R}$.
\end{lem}
\begin{proof}
From the first order condition  $$\underset{t\rightarrow\infty}{\lim}\ \frac{1-F(tx)}{1-F(t)}=x^{-1/\gamma}$$ follows that $U=F^{\leftharpoonup}(1-1/x)$ is regularly varying with index $\gamma$ and there exists a slowly varying function $\ell_U$, such that $U(x)=x^{\gamma} \ell_U(x)$. Let $P_1,P_2,\dots$ be i.i.d.\ random variables with distribution function $1-1/y$. Note that $U(P_i)\Deq X_i$. We define $Y_{(k-i,k)}:=\log(X_{(n-i,n)})-\log(X_{(n-k,n)})$, for which it follows that
\begin{align*}
Y_{(k-i,k)} \Deq \log\left( \frac{U(P_{(n-i,n)})}{U(P_{(n-k,n)})}\right) = \gamma\log\left( \frac{P_{(n-i,n)}}{P_{(n-k,n)}}\right) +\log\left( \frac{\ell_U(P_{(n-i,n)})}{\ell_U(P_{(n-k,n)})}\right).
\end{align*}
Note that $\log(P_i)$ is standard exponentially distributed. By Lemma 3.2.3 in \cite{deHaan-book} follows for i.i.d.\ standard exponential random variables $E_1, E_2,\dots$ that $\left\{ E_{(n-i,n)}-E_{(n-k,n)}\right\}_{i=1}^{k-1} \Deq \left\{ E_{(k-i,k)}\right\}_{i=0}^{k-1}.$\\
Hence, we obtain for the Hill estimator that
\begin{align*}
\hat{\gamma}_k \Deq \gamma \frac{1}{k}\sum_{i=0}^{k-1} E_{(k-i,k)} + \frac{1}{k}\sum_{i=0}^{k-1} \log\left( \frac{\ell_U(P_{(n-i,n)})}{\ell_U(P_{(n-k,n)})}\right)\Deq G_k + b_{n,k},
\end{align*}
where $G_k \sim \Gamma(k,\gamma/k)$ as the sum of i.i.d.\ exponentials and $b_{n,k}$ denotes the second average. \\
If $k/n\rightarrow0$, $P_{(n-k,n)}\rightarrow\infty$ almost surely by Lemma 3.2.1 in \cite{deHaan-book}. Since $\ell_U$ is slowly varying, $b_{n,k}$ converges to zero almost surely.\\
If $k/n\rightarrow c\in(0,1]$, $P_{(n-k,n)}\rightarrow 1/c$ in probability by Cor.\ 2.2.2 in \cite{deHaan-book}. Thus, by the weak law of large numbers
\begin{align*}
b_{n,k} \overset{\mathbb{P}}{\longrightarrow} \mathbb{E}\left[ \log\left( \frac{\ell_U(P)}{\ell_U(1/c)}\right) \Big| P>1/c  \right]=:b_c, \text{ as } n\rightarrow\infty.
\end{align*}
\end{proof}

\begin{lem}\label{L:SKK2_range_k}
Let $X_1,\dots,X_n\iid F$ for $F\in \mathrm{DoA}(G_{\gamma})$ with $\gamma>0$. Then the following holds for $\SKK=(4-k)/(2\hat{\gamma}_k k)$ and $\SKK_{-}=(4+k)/(2\hat{\gamma}_k k)$ depending on the sample fraction $k$.
\begin{enumerate}
\item If k is finite,
\begin{itemize}
\item then $\mathbb{E}[\SKK]+ (2\gamma)^{-1} \rightarrow \frac{3}{2\gamma(k-1)}>0$, as $n\rightarrow\infty$.
\item then $\mathbb{E}[\SKK^{-}]- (2\gamma)^{-1} \rightarrow \frac{5}{2\gamma(k-1)}>0$, as $n\rightarrow\infty$.
\end{itemize} 
\item If $k\rightarrow\infty,\ k/n\rightarrow0$,
\begin{itemize}
\item then $\mathbb{E}[\SKK]+ (2\gamma)^{-1} \rightarrow 0$, as $n\rightarrow\infty$.
\item then $\mathbb{E}[\SKK^{-}]- (2\gamma)^{-1} \rightarrow 0$, as $n\rightarrow\infty$.
\end{itemize} 
\item If $k\rightarrow\infty,\ k/n\rightarrow c>0$
\begin{itemize}
\item and $b_c\geq0$, then $\mathbb{E}[\SKK]+ (2\gamma)^{-1} \rightarrow \frac{b_c}{2\gamma(\gamma+b_c)}>0$, as $n\rightarrow\infty$.
\item and $b_c\leq 0$, then $\mathbb{E}[\SKK^{-}]- (2\gamma)^{-1} \rightarrow \frac{-b_c}{2\gamma(\gamma+b_c)}>0$, as $n\rightarrow\infty$.
\end{itemize} 
\end{enumerate}
\end{lem}
\begin{proof}
We proof these three alternative statements to obtain that the minimizing sequence $k=k_n$ is an intermediate sequence.
\vspace*{-0.2cm}
\begin{enumerate}
\item Let $k$ be finite, by Lemma \ref{L:Hill-dist} holds that $\hat{\gamma}_k\overset{\mathcal{D}}{=} G_k +b_{k,n}$, where $G_k\sim \Gamma(k,\gamma/k)$ and $b_{k,n}\rightarrow 0$ as $n\rightarrow \infty$. Thus, as $n\rightarrow\infty$ it follows that
\begin{align*}
\mathbb{E}[\SKK]+ (2\gamma)^{-1} \rightarrow \mathbb{E}\left[ \frac{4-k}{2G_k k} + \frac{1}{2\gamma} \right]= \frac{3}{2\gamma(k-1)},\\
\mathbb{E}[\SKK^{-}]- (2\gamma)^{-1} \rightarrow \mathbb{E}\left[ \frac{4+k}{2G_k k} - \frac{1}{2\gamma} \right]= \frac{5}{2\gamma(k-1)}.
\end{align*}
\item The statement follows from the consistency of the Hill estimator and the continuous mapping theorem.
\item Let $k\rightarrow\infty$ and $k/n\rightarrow c>0$, then it holds by Lemma \ref{L:Hill-dist} that $\hat{\gamma}_k\overset{\mathcal{D}}{=} G_k +b_{k,n}$, where $G_k\overset{\mathbb{P}}{\rightarrow} \gamma$ and $b_{k,n}\rightarrow b_c$ as $n\rightarrow \infty$. Thus, as $n\rightarrow\infty$
\begin{align*}
\mathbb{E}[\SKK]+ (2\gamma)^{-1}  \rightarrow \frac{-1}{2(\gamma+b_c)} + \frac{1}{2\gamma}  = \frac{b_c}{2\gamma(\gamma + b_c)}, \text{ if } b_c>0,\\
\mathbb{E}[\SKK^{-}]- (2\gamma)^{-1}  \rightarrow \frac{1}{2(\gamma+b_c)} - \frac{1}{2\gamma}  = \frac{-b_c}{2\gamma(\gamma + b_c)}, \text{ if } b_c<0.
\end{align*}
\end{enumerate}
\end{proof}

\begin{lem}\label{L:PQR}
Let $E_1,\dots,E_n$ be i.i.d.\ standard exponential random variables and $k$ s.t.\ $1\leq k\leq n$ and $k\rightarrow\infty$ as $n\rightarrow\infty$. We define the following random variables,
\begin{align*}
&P_{k,n} := \sqrt{k} \left( \frac{1}{k}\sum_{i=1}^k E_i -1 \right), \\
&Q_{k,n} := \sqrt{k} \left( \frac{1}{k}\sum_{i=1}^k E_i^2 -2 \right), \\
&R_{k,n} := \sqrt{k} \left( \frac{1}{k}\sum_{i=1}^k e_{i+1}^k E_i -1 \right), \\
\end{align*}
where $e_i^k:=\sum_{l=i}^k l^{-1}=\mathbb{E}[E_{(k-i+1,k)}]$. Then it holds for $n\rightarrow\infty$ that
\begin{align*}
(P_{k,n}, Q_{k,n}, R_{k,n})^{T} \Darrow \mathcal{N}\left( (0,0,0)^{T}, \begin{pmatrix}
1 & 4 & 1 \\
4 & 20 & 4 \\
1 & 4 & 2 \\
\end{pmatrix}  \right).
\end{align*}
\end{lem}
\begin{proof}
We use the Cram\'er-Wold device that gives us a joint normal limit distribution if all linear combinations have an univariate normal limit distribution. For $a_1, a_2, a_3 \in \mathbb{R}$ we study
\begin{equation}\label{E:lin_comb_PQR}
 \left(  a_1 P_{k,n} + a_2 Q_{k,n} + a_3 R_{k,n}  \right).
\end{equation}
To prove asymptotic normality for the sum in (\ref{E:lin_comb_PQR}) we use Liapounov's central limit theorem (CLT), see Theorem 7.1.2.\ in \cite{Chung-book}. We consider a sum $S_n:=\sum_{i=1}^k X_{i,k}$ of independent random variables fulfilling the following three conditions,
\begin{enumerate}
\item[1)] $\mathbb{E}[X_{i,k}]=0, \ \forall k\ \forall i$,
\item[2)] $\sum_{i=1}^k \mathrm{Var}(X_{i,k}) = \sigma^2 $,
\item[3)] $\Gamma(k)=\sum_{i=1}^k \mathbb{E}[|X_{i,k}|^3] \longrightarrow 0,\ \text{as } k\rightarrow\infty,$
\end{enumerate}
Then the CLT proves a standard normal limit for $S_n$. We define
\begin{equation*}
X_{i,k}:= \frac{1}{\sqrt{k}} \left( a_1 E_i - a_1 + a_2 E_i^2 - 2a_2 + a_3e_{i+1}^k E_i -a_3 e_{i+1}^k \right)
\end{equation*}
for $i=1,\dots,k$ where $e_{k+1}^k:=0$, such that $\sum_{i=1}^k X_{i,k} \approx$(\ref{E:lin_comb_PQR}), where $c_k\approx c$ if $c_k\rightarrow c$ as $k\rightarrow\infty$ and the approximation error is due to
\begin{align*}
\frac{1}{k}\sum_{i=1}^{k-1} e_{i+1}^k = \frac{1}{k}\sum_{i=1}^{k-1} \frac{1}{k}\sum_{l=i+1}^{k} \frac{1}{l/k} \ \approx \ \int_{\frac{1}{k}}^1 \int_v^1 \frac{1}{u} \mathrm{d}u \mathrm{d}v = -\int_{\frac{1}{k}}^1 \log(v)\mathrm{d}v \underset{k\rightarrow\infty}{\rightarrow}  1.
\end{align*}
Now we have to check the three conditions. Condition 1) follows immediately from $\mathbb{E}[E_i]=1$ and $\mathbb{E}[E_i^2]=2$. For condition 2) we need to calculate the variance
\begin{align*}
\mathrm{Var}(X_{i,k}) &= \mathrm{Var}\left( \frac{1}{\sqrt{k}} \left( a_1 E_i - a_1 + a_2 E_i^2 - 2a_2 + a_3e_i^k E_i -a_3 e_i^k \right) \right)\\
&= \frac{1}{k}\Big(  a_1^2\mathrm{Var}(E_i) +a_2^2\mathrm{Var}(E_i^2) + a_3^2(e_{i+1}^k)^2 \mathrm{Var}(E_i)\\
&\qquad  + 2(a_1a_2+a_2a_3e_{i+1}^k) \Cov(E_i,E_i^2)+ 2a_1a_3 \mathrm{Var}(E_i) \Big)\\
&= \frac{1}{k}\Big(  a_1^2 + 20a_2^2 + a_3^2(e_{i+1}^k)^2 + 8(a_1a_2+a_2a_3e_{i+1}^k)+ 2a_1a_3  \Big),
\end{align*}
since $\mathrm{Var}(E_i)=1$, $\mathrm{Var}(E_i^2)=20$ and $\Cov(E_i,E_i^2)=\mathbb{E}[E_i^3]-\mathbb{E}[E_i^2]\mathbb{E}[E_i]=4$. With the approximation
\begin{align*}
\frac{1}{k}\sum_{i=1}^{k-1}( e_{i+1}^k )^2\ \approx \ \int_{\frac{1}{k}}^1 \Big( \int_v^1  \frac{1}{u} \mathrm{d}u \Big)^2 \mathrm{d}v \underset{k\rightarrow\infty}{\rightarrow}  2
\end{align*}
follows that
\begin{equation*}
\sum_{i=1}^k \mathrm{Var}(X_{i,k}) = a_1^2 + 20a_2^2 + 2a_3^2 + 8a_1a_2 + 8a_2a_3 + 2a_1a_3.
\end{equation*}
Condition 3) holds with 
\begin{equation}
\Gamma(k)= \frac{1}{k\sqrt{k}} \sum_{i=1}^k \mathbb{E}\Big[\big|(a_1+a_3e_i^k)(E_i-1)+a_2(E_i^2-2)\big|^3\Big] = \frac{c}{\sqrt{k}}\rightarrow 0,
\end{equation}
as $k\rightarrow\infty$ and for a constant $c>0$, since the exponential distribution has finite moments and $\sum_{i=1}^k (e_{i+1}^k)^3/k \approx 6$. Thus, we obtain that
\begin{equation*}
\sum_{i=1}^k X_{i,k} \Darrow \mathcal{N}(0,\ a_1^2 + 20a_2^2 + 2a_3^2 + 8a_1a_2 + 8a_2a_3 + 2a_1a_3 ).
\end{equation*}
This is the limiting distribution of the sum in (\ref{E:lin_comb_PQR}) and also follows from the joint normal distribution.
\end{proof}

\begin{lem}\label{L:distributional_representaions}
Let $X_1,\dots,X_n\iid F$ for $F\in \mathrm{DoA}(G_{\gamma})$ with $\gamma>0$ and $P_1,P_2,\dots$ be i.i.d.\ random variables with distribution function $1-1/y$. We define
\begin{align*}
\hat{\gamma}_k &:=  \frac{1}{k}\sum_{i=1}^k \log\left(\frac{X_{(n-i+1,n)}}{X_{(n-k,n)}}  \right),  \quad M_n:=  \frac{1}{k}\sum_{i=1}^k \log\left(\frac{X_{(n-i+1,n)}}{X_{(n-k,n)}}  \right)^2, \\
\overline{YE} &:=  \frac{1}{k}\sum_{i=1}^k \log\left(\frac{X_{(n-i+1,n)}}{X_{(n-k,n)}}  \right)  e_{i}^k , \text{ and}\quad e_i^k:=\sum_{l=i}^k \frac{1}{l}.
\end{align*}
If the second order condition
\begin{align*}
\underset{t\rightarrow\infty}{\lim} \frac{\frac{U(tx)}{U(t)}-x^{\gamma}}{A(t)} =x^{\gamma}\frac{x^{\rho}-1}{\rho}
\end{align*}
holds for $\rho<0$ and $x>0$, then
\begin{align}
&\hat{\gamma}_k \Deq \gamma + \gamma P_{k,n}/\sqrt{k} + \frac{A(Y_{(n-k,n)})}{1-\rho} + o_p(A(n/k)),\label{E:Hill-dist}\\
&M_n \Deq 2\gamma^2 + \gamma^2 Q_{k,n}/\sqrt{k} + \frac{2\gamma(2-\rho)}{(1-\rho)^2}A(Y_{(n-k,n)}) + o_p(A(n/k)),\label{E:Mn-dist}\\
&\overline{YE} \Deq 2\gamma + \gamma( P_{k,n} + R_{k,n})/\sqrt{k} + \frac{2-\rho}{(1-\rho)^2}A(Y_{(n-k,n)}) + o_p(A(n/k)).\label{E:YE-dist}
\end{align}
\end{lem}
\begin{proof}
The results in (\ref{E:Hill-dist}) and (\ref{E:Mn-dist}) are already stated in the proof of Theorem 1 in \cite{deHaan-Peng}.\\
To prove (\ref{E:YE-dist}) we follow the proof of the asymptotic normality of the Hill estimator in \cite{deHaan-book}. Let $A_0$ be such that $A(t)/A_0(t)\rightarrow1$, as $t\rightarrow\infty$. Then, for each $\epsilon>0$ there exists a $t_0$ such that for $t\geq t_0$ and $x\geq1$ the inequality in Theorem B.2.18 in \cite{deHaan-book} holds. For $t=P_{n-k,n}$ and $x=P_{(n-i,n)}/P_{(n-k,n)}$ we obtain that
\begin{align*}
\overline{YE} &\Deq \frac{\gamma}{k} \sum_{i=1}^k \log\left( \frac{P_{(n-i+1,n)}}{P_{(n-k,n)}} \right)e_{i}^k + A_0(P_{(n-k,n)}) \frac{1}{k}\sum_{i=1}^k \frac{\left( \frac{P_{(n-i+1,n)}}{P_{(n-k,n)}} \right)^{\rho}-1}{\rho} e_{i}^k\\
&\qquad +o_p(1)|A_0(P_{(n-k,n)})|\frac{1}{k}\sum_{i=1}^k\left( \frac{P_{(n-i+1,n)}}{P_{(n-k,n)}} \right)^{\rho+\epsilon} e_{i}^k.
\end{align*}
The second term can be approximated by 
\begin{align*}
\frac{1}{k}\sum_{i=1}^k \frac{\left( \frac{P_{(n-i+1,n)}}{P_{(n-k,n)}} \right)^{\rho}-1}{\rho} e_{i}^k \underset{}{\rightarrow} \int_0^1 \frac{v^{-\rho}-1}{\rho} \int_v^1 \frac{1}{u} \mathrm{d}u \mathrm{d}v = \frac{2-\rho}{(1-\rho)^2},
\end{align*}
and for the third term holds
\begin{align*}
 \frac{1}{k}\sum_{i=1}^k \left( \frac{P_{(n-i+1,n)}}{P_{(n-k,n)}} \right)^{\rho+\epsilon} e_{i}^k \underset{}{\rightarrow} \int_0^1 v^{-\rho-\epsilon} \int_v^1 \frac{1}{u} \mathrm{d}u \mathrm{d}v = \frac{1}{(1-\rho-\epsilon)^2},
\end{align*}
as $k\rightarrow\infty$.
Note that for $E_1,\dots,E_n$ i.i.d.\ standard exponential random variables follows by R\'enyi's representation that
\begin{align*}
\left\{ \log\left(\frac{P_{(n-i+1,n)}}{P_{(n-k,n)}}  \right) \right\}_{i=1}^k \Deq \left\{ E_{(k-i+1,k)} \right\}_{i=1}^k \Deq \left\{\sum_{j=i}^{k} \frac{E_j}{j} \right\}_{i=1}^k.
\end{align*}
This distributional equality enables the following transformations,
\begin{align*}
 &\sum_{i=1}^k \log\left( \frac{P_{(n-i+1,n)}}{P_{(n-k,n)}} \right)e_{i}^k \Deq  \sum_{i=1}^k  e_{i}^k \sum_{j=i}^{k} \frac{E_j}{j} \\
&\qquad \qquad = \sum_{i=1}^k \frac{E_i}{i} \sum_{j=1}^{i} e_j^k =  \sum_{i=1}^k \frac{E_i}{i} \Big(  i e_{i+1}^k + \sum_{j=1}^i \sum_{l=j}^i \frac{1}{l}\Big)\\
&\qquad \qquad =  \sum_{i=1}^k \frac{E_i}{i} \Big( i+ i e_{i+1}^k\Big)=  \sum_{i=1}^k  E_i +  \sum_{i=1}^k E_i e_{i+1}^k.
\end{align*}
Thus,
\begin{align*}
\frac{\gamma}{k} \sum_{i=1}^k \log\left( \frac{P_{(n-i+1,n)}}{P_{(n-k,n)}} \right)e_{i}^k &\Deq 2\gamma + \gamma \bigg(\Big(\frac{1}{k}  \sum_{i=1}^k  E_i -1 \Big)+\Big( \frac{1}{k}\sum_{i=1}^k E_i e_{i+1}^k  -1 \Big)    \bigg)\\
&=2\gamma +\gamma\left( P_{k,n} + R_{k,n} \right)/\sqrt{k}.
\end{align*}
Combining the above arguments as in the proof of Theorem 3.2.5 in \cite{deHaan-book} gives (\ref{E:YE-dist}).
\end{proof}

\begin{lem}\label{L:Cov_K_k}
For $k\rightarrow\infty$, $k/n\rightarrow0$ and $k/K\rightarrow c$ with $0<c<1$,
\begin{align*}
\Cov(R_{K,n},R_{k,n})\, &\rightarrow \, \frac{2c-c\log(c)}{\sqrt{c}},\\
\Cov(R_{K,n},P_{k,n})\, &\rightarrow \, \frac{c-c\log(c)}{\sqrt{c}}, \text{ as } n \rightarrow\infty,
\end{align*}
where $R_{k,n}$ and $P_{k,n}$ are defined in Lemma \ref{L:PQR}.
\end{lem}
\begin{proof}
Let $E_1, E_2,\dots$ be i.i.d.\ standard exponential random variables, where $\Cov(E_i,E_j)$ is equal to 1 if $i=j$ and 0 otherwise. Then
\begin{align*}
 \Cov(&R_{K,n},R_{k,n}) = \Cov\left( \sum_{i=1}^k E_i \frac{e_{i+1}^k}{\sqrt{k}} , \sum_{i=1}^K E_i \frac{e_{i+1}^K}{\sqrt{K}} \right)\\
 &= \sum_{i=1}^k \sum_{j=1}^K \frac{e_{i+1}^k}{\sqrt{k}} \frac{e_{j+1}^K}{\sqrt{K}} \Cov(E_i, E_j)= \frac{\sqrt{k}}{\sqrt{K}} \frac{1}{k} \sum_{i=1}^k  \bigg( \frac{1}{k}\sum_{l=i+1}^k \frac{1}{l/k} \bigg) \bigg( \frac{1}{K}\sum_{l=i+1}^K \frac{1}{l/K} \bigg)\\
 &\approx \frac{\sqrt{k}}{\sqrt{K}} \frac{1}{k} \sum_{i=1}^k \bigg( \int_{(i+1)/k}^1  \frac{1}{u}\mathrm{d}u \bigg)  \bigg( \int_{(i+1)/K}^1  \frac{1}{u}\mathrm{d}u \bigg)\\
 & \approx\sqrt{c} \int_0^1 \log(v)^2-\log(v)\log(c) \mathrm{d}v
 = 2\sqrt{c} -\sqrt{c}\log(c),
\end{align*} 
where $c_k\approx c$ again denotes that $c_k\rightarrow c$ as $k\rightarrow\infty$.\\
In the same way we obtain
\begin{align*}
\Cov(&R_{K,n},P_{k,n}) = \sum_{i=1}^k \sum_{j=1}^K \frac{1}{\sqrt{k}} \frac{e_{j+1}^K}{\sqrt{K}} \Cov(E_i, E_j) \\
&\approx  \frac{\sqrt{k}}{\sqrt{K}} \frac{1}{k} \sum_{i=1}^k  \bigg( \int_{(i+1)/K}^1  \frac{1}{u}\mathrm{d}u \bigg) \approx \frac{\sqrt{k}}{\sqrt{K}} \int_{0}^1 \bigg( \int_{cv}^1 \frac{1}{u}\mathrm{d}u \bigg) \mathrm{d}v\\
 &= \sqrt{c} -\sqrt{c}\log(c)
\end{align*}
\end{proof}

\begin{thm}\label{T:averageHill}
Let $\bar{\gamma}_k := \frac{1}{k}\sum_{i=1}^k  \hat{\gamma}_i$ denote the average over Hill estimates. Further let $X_1,\dots,X_n$ be i.i.d.\ random variables with distribution function $F\in \mathrm{DoA}(G_{\gamma})$, $\gamma>0$. If 
\begin{equation}
\underset{t\rightarrow\infty}{\lim} \frac{\frac{U(tx)}{U(t)}-x^{\gamma}}{A(t)} =x^{\gamma}\frac{x^{\rho}-1}{\rho},
\end{equation}
with $U(x):=F^{\leftharpoonup}\left(1-\frac{1}{x}\right)$ holds and $k\rightarrow\infty$ and $k/n\rightarrow0$ as $n\rightarrow\infty$,
\begin{align*}
\sqrt{k} \left(\bar{\gamma}_{k}-\gamma   \right) \overset{\mathcal{D}}{\longrightarrow} \mathcal{N}\left( \frac{\lambda}{(1-\rho)^2}, 2\gamma^2 \right).
\end{align*}
with $\lambda:= \underset{k\rightarrow\infty}{\lim} \sqrt{k}A(n/k)$. 
\end{thm}
\begin{proof}
First we have to rewrite the average over the Hill estimator,
\begin{align*}
\bar{\gamma}_k &= \frac{1}{k}\sum_{i=1}^k \hat{\gamma}_k =\frac{1}{k}\sum_{i=1}^k \frac{1}{i}\sum_{j=1}^i \log\left( \frac{X_{(n-j+1,n)}}{X_{(n-i,n)}} \right)\\
&=\frac{1}{k}\sum_{i=1}^k \log X_{(n-i+1,n)}\sum_{j=i}^k \frac{1}{j}  -\frac{1}{k}\sum_{i=1}^k \log X_{(n-i,n)}\\
&= \frac{1}{k}\sum_{i=1}^k \log\left( \frac{X_{(n-i+1,n)}}{X_{(n-k,n)}} \right)\sum_{j=i}^k \frac{1}{j}  -\frac{1}{k}\sum_{i=1}^k \log\left( \frac{X_{(n-i,n)}}{X_{(n-k,n)}} \right)\\
&= \overline{YE}- \hat{\gamma}_k +\frac{1}{k}\log\left( \frac{X_{(n,n)}}{X_{(n-k,n)}} \right),
\end{align*}
where $\overline{YE}$ is defined in Lemma \ref{L:distributional_representaions}.
Following the proof of Lemma \ref{L:distributional_representaions} it holds that the last term above is in distribution equal to
\begin{align*}
\frac{\gamma}{k}\log\left( \frac{P_{(n,n)}}{P_{(n-k,n)}} \right) + \frac{A(P_{(n-k,n)})}{k}\frac{\left( \frac{P_{(n,n)}}{P_{(n-k,n)}} \right)^{\rho}-1}{\rho} + o_p(1)\frac{|A(P_{(n-k,n)})|}{k}\left( \frac{P_{(n,n)}}{P_{(n-k,n)}} \right)^{\rho+\epsilon}.
\end{align*} 
From Corollary 2.2.2 in \cite{deHaan-book} follows that 
$$\frac{k}{n}P_{(n-k,n)}\Parrow 1, \text{ and }\ \frac{P_{(n,n)}}{k P_{(n-k,n)}}\Parrow 1,\text{ as } n\rightarrow\infty.$$ 
Thus, $ \big(\log(X_{(n,n)})-\log(X_{(n-k,n)})\big)/k= O_p(\log(k)/k)+ O_p(A(n/k)/k)$, and by Lemma \ref{L:distributional_representaions} follows
\begin{equation*}
\bar{\gamma}_k \Deq \gamma + \gamma R_{k,n}/\sqrt{k} + \frac{A(n/k)}{(1-\rho)^2} + o_p(A(n/k))+ O_p(\log(k)/k).
\end{equation*}
\end{proof}

\begin{thm}\label{T:upperMean}
Let $\bar{\gamma}_{\mathrm{u}p,K,k} := \frac{1}{K-k+1}\sum_{i=k}^K \hat{\gamma}_i$ be the upper mean and $k$ and $K$ intermediate sequences, i.e.\ $k\rightarrow\infty$ and $k/n\rightarrow0$, as $n\rightarrow\infty$. Further, let $\sqrt{k}A(n/k)\rightarrow\lambda$ and $k/K\rightarrow c$ with $0<c<1$. 
Under the conditions of Theorem \ref{T:averageHill}, it holds that \begin{equation*}
\sqrt{k}(\bar{\gamma}_{\mathrm{up},K,k} -\gamma) \Darrow \mathcal{N}\left(\frac{\lambda}{(1-\rho)^2}\frac{c^{\rho}-c}{1-c}, \frac{2\gamma^2c}{1-c}\Big(1+\frac{c\log(c)}{1-c}\Big) \right).
\end{equation*}
\end{thm}
\begin{proof}
We can write the upper mean as a combination of two averaged Hill estimators and apply Theorem \ref{T:averageHill},
\begin{align*}
\bar{\gamma}_{\mathrm{up},K,k}&= \frac{K}{K-k+1}\bar{\gamma}_K - \frac{k}{K-k+1}\bar{\gamma}_k\\
&\Deq \gamma + \gamma\frac{K}{K-k+1} R_{K,n}/\sqrt{K} - \gamma\frac{k}{K-k+1} R_{k,n}/\sqrt{k} \\
&\qquad + \frac{K}{K-k+1}\frac{A(n/K)}{(1-\rho)^2}- \frac{k}{K-k+1}\frac{A(n/k)}{(1-\rho)^2} + o_p(A(n/K)).
\end{align*}
We approximate $k/K$ by $c$ and obtain
\begin{align*}
\sqrt{k}(\bar{\gamma}_{\mathrm{up},K,k}-\gamma) &\Deq \gamma\frac{\sqrt{c}}{1-c} R_{K,n} - \gamma\frac{c}{(1-c)} R_{k,n} \\
&\qquad + \frac{1}{1-c}\frac{\sqrt{k}A(n/K)}{(1-\rho)^2}- \frac{c}{(1-c)}\frac{\sqrt{k}A(n/k)}{(1-\rho)^2} + o_p(1).
\end{align*}
Now we need the covariance between $R_{n,k}$ and $R_{n,K}$, see Lemma \ref{L:Cov_K_k}, and apply the following property of regular varying functions,
\begin{align*}
\sqrt{k}A(n/K)= \sqrt{k}A(n/k)\frac{A(cn/k)}{A(n/k)}\rightarrow \lambda c^{\rho},\ \text{as } n \rightarrow\infty.
\end{align*}
This leads to
\begin{align*}
&\sqrt{K}(\bar{\gamma}_{\mathrm{up},K,k}-\gamma)\\
&\qquad \Darrow \mathcal{N}\left(\frac{\lambda}{(1-\rho)^2}\frac{c^{\rho}-c}{1-c}, \frac{\gamma^2c}{(1-c)^2}\Big(2+2c-2\sqrt{c}\Big( \frac{2c-c\log(c)}{\sqrt{c}}\Big)\Big) \right).
\end{align*}
\end{proof}

\begin{proof}[Proof of Theorem \ref{T:Bias}]
The bias estimator is defined as $\bar{b}_{\mathrm{up},K,k}= \bar{\gamma}_{\mathrm{up},K,k}-\bar{\gamma}_K$ in equation (\ref{E:bias_est}). Thus, we can utilize the asymptotic normality results for $\bar{\gamma}_k$ and $\bar{\gamma}_{\mathrm{up},K,k}$ in Theorem \ref{T:averageHill} and \ref{T:upperMean}. Following the proofs of these theorems it holds that
\begin{align*}
\sqrt{k}&\, \bar{b}_{\mathrm{up},K,k} \Deq \gamma\frac{k-1}{K-k+1}\frac{\sqrt{k}}{\sqrt{K}} R_{K,n} -\gamma\frac{k}{K-k+1}R_{k,n} \\
& \quad+ \frac{k}{K-k+1}\frac{\sqrt{k}(A(n/K)-A(n/k))}{(1-\rho)^2} + \sqrt{k}\Big(o_p(A(n/K))+ o_p(A(n/k))\Big).
\end{align*}
Here the random variable $R_{k,n}$ is defined in Lemma \ref{L:PQR} and we know that $R_{k,n}$ has a normal limit distribution. With Lemma \ref{L:PQR} and Lemma \ref{L:Cov_K_k} we obtain the following variance,
\begin{align*}
& \mathrm{Var}\bigg(\gamma\frac{k-1}{K-k+1}\frac{\sqrt{k}}{\sqrt{K}} R_{K,n} -\gamma\frac{k}{K-k+1}R_{k,n}  \bigg)\\
&\ \approx \gamma^2\bigg(  \Big( \frac{c\sqrt{c}}{1-c} \Big)^2\Var(R_{K,n})+  \Big( \frac{c}{1-c} \Big)^2\Var(R_{k,n})- 2\frac{c^2\sqrt{c}}{(1-c)^2}\Cov(R_{n,K},R_{n,k})  \bigg)\\
&\ \approx\gamma^2\frac{2c^3+ 2c^2 -4c^3+ 2c^3\log(c)}{(1-c)^2} = \frac{2\gamma^2c^2}{1-c} \bigg( 1+\frac{c\log(c)}{1-c}\bigg),
\end{align*} 
The bias term of the normal limit is
\begin{align*}
\frac{k}{K-k+1}\frac{\sqrt{k}(A(n/K)-A(n/k))}{(1-\rho)^2}  \rightarrow \frac{\lambda}{(1-\rho)^2} \frac{c(c^{\rho}-1)}{1-c},
\end{align*}
as $n\rightarrow\infty$, which follows from the regular variation of $A$ and due to $\sqrt{k}A(n/k)\rightarrow\lambda$. Since 
\begin{align*}
\sqrt{k}\Big(o_p(A(n/K))+ o_p(A(n/k))\Big)=o_p(1),
\end{align*}
the statement of the theorem follows immediately.
\end{proof}

\section*{Acknowledgement}
Support of the DFG RTG 2088 (B4) is gratefully acknowledged. We are grateful to Julien Hambuckers for providing us the dataset of operational losses from UniCredit.

\bibliographystyle{chicagoa}    
\bibliography{bibhd} 

\end{document}